\pdfoutput=1
\documentclass{article}
\usepackage{arxiv}

\usepackage{graphicx}
\usepackage{hyperref}
\usepackage[utf8]{inputenc}
\usepackage[T1]{fontenc}
\usepackage{url}
\usepackage{amsfonts}
\usepackage{microtype}
\usepackage{mathptmx}
\usepackage{amsmath}
\usepackage{amsthm}
\usepackage{amssymb}

\newcommand{\doiurl}[1]{\href{https://doi.org/#1}{\nolinkurl{https://doi.org/#1}}}

\newtheoremstyle{thm-break}
  {10pt}
  {3pt}
  {\normalfont}
  {}
  {\bfseries}
  {}
  {\newline}
  {\thmname{#1}\thmnumber{ #2}\thmnote{ (#3)}\par\vspace{0.5\baselineskip}}

\theoremstyle{thm-break}
\newtheorem{theorem}{Theorem}[section]
\theoremstyle{remark}
\newtheorem{remark}[theorem]{Remark}
\theoremstyle{thm-break}
\newtheorem{definition}[theorem]{Definition}
\theoremstyle{thm-break}
\newtheorem{proposition}[theorem]{Proposition}
\theoremstyle{thm-break}
\newtheorem{lemma}[theorem]{Lemma}
\theoremstyle{thm-break}
\newtheorem{corollary}[theorem]{Corollary}
\theoremstyle{remark}
\newtheorem{example}[theorem]{Example}
\theoremstyle{thm-break}

\renewenvironment{proof}[1][Proof]{\par\noindent\textbf{#1. }\normalfont}{\hfill$\square$\par}

\newcommand{\llbracket}{[\![}
\newcommand{\rrbracket}{]\!]}
\title{QBism Logic}
\author{Kenji Tokuo\\
Department of Information Engineering, Oita College\\
National Institute of Technology, Japan\\
\texttt{tokuo@oita-ct.ac.jp}}
\date{July 19, 2026}

\begin{document}
\maketitle

\begin{abstract}
QBism interprets quantum theory as a normative discipline for an agent's probability assignments and their revision across possible experience. This paper develops a logical formalization of that picture. A well-formed core datum consists of an admissible prior space, a finite family of actual measurements, Born kernels, and update kernels. For each such datum, we introduce a guarded dynamic language for histories and posterior states and prove a global reduction theorem. We next consider effectively semialgebraic data over an effectively presented real closed field. For data in this class, we translate the fragment without dynamic operators into first-order formulas in the corresponding language of ordered rings, thereby reducing validity to first-order reasoning over real closed fields. Together, the reduction and first-order translation yield a sound and complete recursive calculus and a decision procedure for validity. Finally, assuming a symmetric informationally complete (SIC) reference measurement, we show that quantum theory in finite dimensions realizes the framework through SIC coordinates, POVMs, and quantum instruments. We also prove that the corresponding SIC image satisfies the standard qplex geometry conditions, namely the consistency bounds and the lower polar condition, and that under explicit coefficient field hypotheses the resulting quantum datum is effectively semialgebraic.
\end{abstract}


\section{Introduction}\label{sec:introduction}

\subsection{Background}\label{subsec:introduction-background}

QBism presents quantum theory as a normative framework for an agent's probability assignments concerning the consequences of prospective actions. From this perspective, a quantum state functions as a compact encoding of personal expectations for future experience, and the Born rule functions as an additional normative constraint on probability assignments beyond ordinary probabilistic coherence \cite{Fuchs2017,FuchsSchack2013}.

A particularly useful form of this idea appears in the \emph{symmetric informationally complete (SIC) representation}. Relative to a fixed reference measurement, the Born rule can be written as an affine relation between a prior distribution on the outcomes of the reference measurement and conditional probability assignments for actual measurements. This representation makes the normative character of the Born rule especially transparent and has motivated extensive research on the geometry of admissible prior spaces, including \emph{qplex geometry} \cite{ApplebyFuchsStaceyZhu2017,FuchsSchack2013}.

This line of work brings several structural ingredients into close contact. A fixed reference measurement supplies a coordinate system for probability assignments, actual measurements supply the interventions whose outcomes may be experienced, and the Born rule links these two levels through an affine constraint. The qplex perspective adds a geometric account of the admissible prior space inside the probability simplex, while the Hilbert space formalism supplies the standard quantum realization of this probabilistic structure \cite{ApplebyFuchsStaceyZhu2017}.

\subsection{Aim}\label{subsec:introduction-aim}

The aim of the present paper is to give a logical formalization of QBism in which the internal structure of the QBist picture becomes explicit and mathematically tractable.

The framework is built from a \emph{well-formed core datum} $(Q,\mathcal D,B,U)$, where $Q$ is an \emph{admissible prior space}, $\mathcal D$ is a finite family of \emph{actual measurements}, $B$ is the family of \emph{Born kernels}, and $U$ is the family of \emph{update kernels}. Once such a datum is fixed, it determines the intended state space, the interpretation of the nonlogical vocabulary, and the corresponding notion of validity. The paper first develops the general theory over arbitrary well-formed core data and then studies quantum realizations in SIC coordinates.

A logical formalization of QBism must do at least four things.

First, it must represent the asymmetry between a fixed reference measurement and the family of actual measurements whose outcomes can occur in experience. In QBist terms, the reference side supplies the coordinate system for present commitment, while the actual side supplies the interventions and recorded outcomes that determine subsequent updates \cite{Fuchs2023,FuchsSchack2013}.

Second, it must represent posterior commitment explicitly. The Born rule connects a current reference prior to a current distribution over actual outcomes. A logical framework for QBism must also represent the revision of commitments after an outcome occurs \cite{Fuchs2023}. This requires histories, posterior reference priors, and a semantics that distinguishes executable from nonexecutable branches.

Third, it must specify where geometric conditions become relevant. The SIC and qplex literature suggests that admissible prior spaces have structured geometry inside the simplex. A formal treatment should distinguish basic semantic stability conditions from stronger geometric conditions that appear in the QBist literature on SIC coordinates \cite{ApplebyFuchsStaceyZhu2017}.

Fourth, it must explain the relation between the resulting structure and standard quantum theory. In the QBist setting, density operators, POVMs, and quantum instruments provide a Hilbert space realization of the probabilistic framework.

\subsection{Contributions}\label{subsec:introduction-contributions}

The paper makes three logically distinct contributions.

At the semantic level, we define the notion of a well-formed core datum. Its definition requires the Born and update kernels to satisfy their positivity and normalization conditions, the affine Born expressions to produce genuine probability distributions on the admissible prior space, and updates following outcomes of positive probability to remain inside the same space. These conditions support the definitions of histories, posterior reference priors, and guarded dynamic operators.

At the metatheoretic level, the central structural result is a \emph{global reduction theorem} for the \emph{guarded dynamic language}. It shows that every formula with dynamic operators is equivalent to a formula in the \emph{fragment without dynamic operators}. For \emph{effectively semialgebraic data}, the fragment without dynamic operators then admits a direct first-order translation into the theory of real closed fields. This reduces the validity problem to effective first-order reasoning over real closed fields. The recursive soundness, completeness, and decidability results arise from that reduction and translation.

At the realization level, whenever a SIC exists in the chosen dimension, standard quantum theory in finite dimensions yields a well-formed core datum with concrete Born kernels, concrete update kernels, and an admissible prior space given by SIC coordinates of density operators. The effective bridge theorem establishes that, under explicit coefficient field hypotheses on the SIC, the POVMs, and the instruments, the resulting quantum datum is effectively semialgebraic. The same analysis proves that the SIC image satisfies the \emph{consistency bounds} and the \emph{lower polar condition} familiar from the qplex literature \cite{ApplebyFuchsStaceyZhu2017}.

\subsection{Guarded dynamic framework}\label{subsec:introduction-guarded}

A distinctive feature of the paper is the use of a guarded dynamic framework. This choice follows from the role played by histories and updating in the intended interpretation.

In the present setting, a history is a finite sequence of actual outcomes, and the corresponding updates are evaluated recursively relative to the current admissible prior. The key semantic issue is the executability of those updates at the state under consideration. Branches with positive probability generate posterior states, whereas a branch with zero probability prevents the update sequence from being extended along that branch. The use of guarded dynamic operators follows the dynamic logic tradition \cite{HarelKozenTiuryn2000,Pratt1976}, and probabilistic dynamic logic provides relevant background for combining modal structure with probabilistic semantics \cite{FeldmanHarel1984,Kozen1985}. The present framework adapts those ideas to histories of QBist updating.

\subsection{Organization}\label{subsec:introduction-organization}

Section~\ref{sec:preliminaries} introduces the geometric and dynamic preliminaries. Section~\ref{sec:semantic-data} defines well-formed core data. Section~\ref{sec:guarded-dynamic-language} develops histories, recursive updates, the language, and the semantics. Section~\ref{sec:metatheory} proves the global reduction theorem, translates the fragment without dynamic operators into first-order formulas over real closed fields, and derives the main effective metatheoretic consequences. Section~\ref{sec:quantum-realization} presents the quantum realization, establishes that the SIC image satisfies the qplex geometry conditions, and proves the effective bridge theorem. Finally, Section~\ref{sec:concluding-remarks} concludes the paper.

\section{Preliminaries}\label{sec:preliminaries}

\subsection{Consistency bounds on the probability simplex}\label{subsec:preliminaries-consistency-bounds}

Throughout the paper, we fix a dimension parameter $d \ge 2$ and write $I := \{1,\dots,d^2\}$. We also use the constants $L_d := 1/(d(d+1))$ and $U_d := 2/(d(d+1))$. The general framework fixes both a reference outcome set of size $d^2$ and this normalization so that the abstract semantics and the later SIC realization use the same coordinates. At this stage, $d$ serves as a fixed structural parameter. The existence of a SIC is required only in Section~\ref{sec:quantum-realization}.

The basic geometric arena is the \emph{probability simplex} on the reference outcome set $I$, namely $\Delta(I) := \{h = (h_i)_{i \in I} \in \mathbb{R}^{d^2} \mid h_i \ge 0 \text{ for all } i \in I \text{ and } \sum_{i \in I} h_i = 1\}$. Elements of $\Delta(I)$ will be called \emph{reference priors}. In the intended QBist interpretation, a point $h \in \Delta(I)$ represents an agent's prior probability assignment for the outcomes of a fixed reference measurement \cite{FuchsSchack2013}.

A central role is played by the \emph{consistency bounds} $L_d \le \langle h,s\rangle \le U_d$ for pairs of priors $h,s \in \Delta(I)$. Here $\langle h,s\rangle := \sum_{i \in I} h_i s_i$ is the standard Euclidean inner product on $\mathbb{R}^{d^2}$. In the SIC and qplex setting, these bounds express a characteristic compatibility constraint on probability vectors \cite{ApplebyFuchsStaceyZhu2017}. In Section~\ref{sec:quantum-realization}, the same bounds will reappear in the geometric description of the SIC image of quantum state space.

For later use, we also introduce the \emph{affine coordinates} $\alpha_i(h) := (d+1)h_i - 1/d$ for $h \in \Delta(I)$ and $i \in I$. These quantities may take negative values. They provide the coefficients that appear naturally in the SIC representation of the Born rule \cite{FuchsSchack2013}. This notation will be used systematically in the affine Born and update expressions, the reduction formulas, and the quantum realization formulas in Sections~\ref{sec:semantic-data}, \ref{sec:metatheory}, and \ref{sec:quantum-realization}.

\subsection{Lower polarity in qplex geometry}\label{subsec:preliminaries-lower-polarity}

Given a subset $Q \subseteq \Delta(I)$, we define its \emph{lower polar} by $Q^\sharp := \{x \in \Delta(I) \mid L_d \le \langle x,s\rangle \text{ for all } s \in Q\}$. The lower polar depends only on the lower consistency bound and is the polarity notion used in the qplex setting under the normalization adopted here \cite{ApplebyFuchsStaceyZhu2017}. Informally, a set satisfies the \emph{lower polar condition} when it already contains exactly those points of the simplex that remain compatible, in the sense of the lower bound, with all of its members.

In this paper, the lower polar condition is used to describe the geometry of the realization in SIC coordinates discussed in Section~\ref{sec:quantum-realization}. In that realization, this polarity condition expresses one of the geometric features satisfied by the SIC image.

\subsection{Dynamic notation for histories}\label{subsec:preliminaries-dynamic}

The dynamic component of the theory concerns finite sequences of actual measurement outcomes. Subsection~\ref{subsec:histories} gives the formal definitions of histories, the empty history $\varepsilon$, extension by an outcome $(D,j)$, and the last record map $\operatorname{last}$. The marker $\bot$ is reserved for the absence of a preceding outcome.

Later sections associate to each admissible prior $h$ and each history $\pi$ a history probability $r_\pi(h)$, a posterior reference prior $\operatorname{up}_{\pi}(h)$ when defined, posterior coordinates $h_i^\pi(h)$, and posterior outcome probabilities $q_j^{D,\pi}(h)$. These quantities are induced by recursive updating from the kernels introduced in Section~\ref{sec:semantic-data}.

The logical language also contains dynamic operators indexed by histories. The formula $[\pi]\varphi$ states that $\varphi$ holds after the history $\pi$. The dual possibility operator is defined by $\langle \pi \rangle \varphi := \neg[\pi]\neg\varphi$. The intended semantics is a \emph{guarded partial correctness semantics}, in the sense familiar from dynamic logic. On a history that is not executable at a given state, the corresponding formula $[\pi]\varphi$ holds trivially. Section~\ref{sec:guarded-dynamic-language} gives the exact semantic clauses.

\section{Semantic data}\label{sec:semantic-data}

\subsection{Measurement framework}\label{subsec:semantic-data-measurement-framework}

The nonlogical data begin with the reference outcome set $I=\{1,\dots,d^2\}$ from Section~\ref{sec:preliminaries}. All admissible priors are indexed by this set. The framework is therefore organized around a single reference measurement at the normative level, while the available actual measurements may vary.

The actual measurement side is specified as follows.

\begin{definition}[Measurement framework]\label{def:measurement-framework}
Let $\mathcal{D}$ be a finite family of actual measurements. Each $D \in \mathcal{D}$ is equipped with a finite nonempty outcome set $O_D$. The pair $(\mathcal{D},(O_D)_{D \in \mathcal{D}})$ is the \emph{measurement framework}.
\end{definition}

The measurement framework is treated abstractly before any Hilbert space realization is imposed. An element $D \in \mathcal{D}$ is an available actual measurement type, and an element $j \in O_D$ is one of its possible recorded outcomes.

The reference measurement serves as the fixed index system for the agent's prior probabilities. The actual measurements in $\mathcal{D}$ supply the interventions that may occur in histories. The framework thereby represents present commitments through the reference measurement and subsequent updating through actual experience.

\subsection{Kernel structure}\label{subsec:semantic-data-kernels}

For each actual measurement $D \in \mathcal{D}$, we specify a \emph{Born kernel} $B^D(j|i) \in \mathbb{R}$ for $i \in I$ and $j \in O_D$. The kernel values are required to satisfy $B^D(j|i) \ge 0$ and $\sum_{j \in O_D} B^D(j|i)=1$ for every $i \in I$. Thus, for each fixed $i$, the map $j \mapsto B^D(j|i)$ is a probability distribution on $O_D$. Intuitively, $B^D(j|i)$ expresses the conditional weight assigned to the actual outcome $j$ of $D$ relative to the reference index $i$. At the abstract level, these kernels are primitive data. The quantum realization in Section~\ref{sec:quantum-realization} will relate them to standard quantum measurement theory, for which Davies--Lewis instruments and Holevo's treatment provide background \cite{DaviesLewis1970,Holevo2012}.

The corresponding actual outcome probabilities are computed from a reference prior $h \in \Delta(I)$ by the \emph{affine Born expression} $q_j^D(h):=\sum_{i \in I}\alpha_i(h)B^D(j|i)$. The coefficients $\alpha_i(h)$ may be negative. The admissible prior space is therefore restricted to priors for which every $q_j^D(h)$ is nonnegative. This is a normative constraint on admissible priors.

For each pair $(D,j)$ with $D \in \mathcal{D}$ and $j \in O_D$, we also specify an \emph{update kernel} $U^{D,j}(m|i) \in \mathbb{R}$ for $i,m \in I$. The kernel values are required to satisfy $U^{D,j}(m|i) \ge 0$ for all $i,m \in I$, and $\sum_{m \in I}U^{D,j}(m|i)=B^D(j|i)$ for every $i \in I$. Thus, for fixed $i$, the map $m \mapsto U^{D,j}(m|i)$ is a subprobability distribution on the reference index set whose total weight is exactly the conditional weight of the actual outcome $j$. This makes $U^{D,j}$ the kernel that controls posterior reweighting following the outcome $(D,j)$. These update data also accord with recent QBist discussions of dynamics \cite{DeBrotaFuchsSchack2024}.

For each reference prior $h \in \Delta(I)$ such that $q_j^D(h)>0$, define the \emph{one-step update vector} $\operatorname{up}_{D,j}(h)$ by the formula $\operatorname{up}_{D,j}(h)_m:=(\sum_{i \in I}\alpha_i(h)U^{D,j}(m|i))/q_j^D(h)$. Here the numerator is affine in the reference prior, while division by $q_j^D(h)$ provides the normalization. The closure condition introduced below ensures that, for admissible priors, this vector belongs to the same admissible prior space and therefore serves as the posterior reference prior.

\subsection{Well-formed core data}\label{subsec:semantic-data-well-formed}

The remaining component is a designated set $Q \subseteq \Delta(I)$ of \emph{admissible priors}. The role of $Q$ is to collect those priors for which the affine Born expressions define genuine probability distributions for all actual measurements in $\mathcal{D}$ and for which the corresponding one-step updates remain inside the same state space. Thus the datum of the theory is the quadruple $(Q,\mathcal{D},B,U)$, where $B$ abbreviates the family of Born kernels and $U$ abbreviates the family of update kernels.

The following definition assembles these ingredients under the conditions required for histories and guarded dynamic operators.

\begin{definition}[Well-formed core datum]\label{def:well-formed-core-datum}
Let $(\mathcal{D},(O_D)_{D \in \mathcal{D}})$ be a measurement framework, and let $B$ and $U$ be families of Born and update kernels satisfying the conditions in Subsection~\ref{subsec:semantic-data-kernels}. The quadruple $(Q,\mathcal{D},B,U)$ is a \emph{well-formed core datum} if the following conditions hold.
\begin{enumerate}
\item $Q$ is a nonempty subset of $\Delta(I)$.
\item For every $h \in Q$ and every $D \in \mathcal{D}$, the numbers $q_j^D(h)$ form a probability distribution on $O_D$.
\item Whenever $h \in Q$, $D \in \mathcal{D}$, $j \in O_D$, and $q_j^D(h)>0$, the one-step update vector $\operatorname{up}_{D,j}(h)$ belongs to $Q$.
\end{enumerate}
\end{definition}

For $h \in \Delta(I)$, normalization in clause (ii) follows from the kernel normalization and the identity $\sum_{i \in I}\alpha_i(h)=1$. Thus clause (ii) requires the affine Born expressions to be nonnegative on $Q$ and states this requirement in the probability distribution form used by the later dynamic clauses.

\medskip

\begin{example}[Well-formed core datum on the full simplex]\label{ex:full-simplex-datum}
Let $Q:=\Delta(I)$. For each $D \in \mathcal D$, define $B^D(j|i):=1/|O_D|$ and $U^{D,j}(m|i):=1/(d^2|O_D|)$. Since $\sum_{i \in I}\alpha_i(h)=1$ for every $h \in \Delta(I)$, one has $q_j^D(h)=1/|O_D|>0$. Moreover, every update gives the uniform prior $\operatorname{up}_{D,j}(h)=(1/d^2,\dots,1/d^2)$. Hence $(\Delta(I),\mathcal D,B,U)$ is a well-formed core datum in which every outcome is executable at every prior. Its admissible prior space does not satisfy the lower consistency bound, because distinct vertices $e_i,e_k\in\Delta(I)$ have $\langle e_i,e_k\rangle=0<L_d$. This datum therefore illustrates well-formedness without the additional qplex geometry satisfied by the admissible prior space in the SIC quantum realization.
\end{example}

\medskip

\begin{remark}[Layers of assumptions]
Well-formedness imposes the stability conditions required by the dynamic semantics and the global reduction theorem. The effective metatheory adds the semialgebraic assumptions introduced in Subsection~\ref{subsec:effectively-semialgebraic-data}. The SIC quantum realization developed in Section~\ref{sec:quantum-realization} is shown to satisfy the qplex geometry conditions.
\end{remark}

\section{Guarded dynamic language}\label{sec:guarded-dynamic-language}

\subsection{Histories}\label{subsec:histories}

Histories connect admissible priors with recorded outcomes. The framework covers individual actual measurements and finite histories of their outcomes, together with the posterior reference priors they generate.

\begin{definition}[Histories]\label{def:histories}
A \emph{history} is a finite sequence of actual outcomes generated by the grammar $\pi ::= \varepsilon \mid \pi;(D,j)$, where $\varepsilon$ is the empty history, $D \in \mathcal{D}$, and $j \in O_D$.
\end{definition}

The expression $\pi;(D,j)$ denotes the history obtained by appending the outcome $j$ of measurement $D$ to $\pi$.

\begin{definition}[Last record map]\label{def:last-record-map}
The \emph{last record map} is specified by $\operatorname{last}(\varepsilon):=\bot$ and $\operatorname{last}(\pi;(D,j)):=(D,j)$.
\end{definition}

The distinguished marker $\bot$ represents the absence of any prior recorded outcome. The function $\operatorname{last}$ will be used later in the semantic clause for dynamic formulas and in the record component of logical states.

Executability is always relative to a chosen admissible prior. A history is \emph{executable} at that prior precisely when the successive one-step updates along the history are defined. For this reason, histories are introduced first as syntactic objects and are paired with history probabilities and posterior reference priors only after the recursive update mechanism has been defined. The distinction between a syntactic history and an executable history is central to the guarded semantics adopted below.

\subsection{Recursive updates}\label{subsec:recursive-updates}

Let $Q$ be the admissible prior space fixed by a well-formed core datum.

\begin{definition}[Recursive updates along histories]\label{def:recursive-updates}
The \emph{posterior reference prior} $\operatorname{up}_{\pi}(h)$ associated with an admissible prior $h \in Q$ and a history $\pi$ is defined by recursion on $\pi$ as follows.
\begin{itemize}
\item $\operatorname{up}_{\varepsilon}(h):=h$.
\item $\operatorname{up}_{\pi;(D,j)}(h):=\operatorname{up}_{D,j}(\operatorname{up}_{\pi}(h))$, when $\operatorname{up}_{\pi}(h)$ is defined and $q_j^D(\operatorname{up}_{\pi}(h))>0$.
\item $\operatorname{up}_{\pi;(D,j)}(h)$ is undefined otherwise.
\end{itemize}
\end{definition}

These recursive clauses reflect the intended QBist interpretation of updating across possible experience. A posterior reference prior after a history is obtained by iterating one-step updates along the history, provided that each successive branch is executable. The positivity condition is essential. Recursive updating proceeds along branches with positive probability, while guarded semantics handles branches with zero probability.

The closure requirement built into well-formedness guarantees that whenever a one-step update is defined, its value remains inside $Q$. The next lemma states the corresponding closure fact for recursively defined posterior reference priors along histories.

\begin{lemma}[Closure under recursive updates]\label{lem:history-closure}
Let $(Q,\mathcal D,B,U)$ be a well-formed core datum. For every admissible prior $h \in Q$ and every history $\pi$, if $\operatorname{up}_{\pi}(h)$ is defined, then $\operatorname{up}_{\pi}(h) \in Q$.
\end{lemma}

\begin{proof}
We argue by induction on the length of $\pi$.

Case $\pi \equiv \varepsilon$. One has $\operatorname{up}_{\varepsilon}(h)=h \in Q$.

Case $\pi \equiv \pi';(D,j)$. Assume the claim holds for $\pi'$ and that $\operatorname{up}_{\pi';(D,j)}(h)$ is defined. By definition of recursive updates, this implies that $\operatorname{up}_{\pi'}(h)$ is defined and that $q_j^D(\operatorname{up}_{\pi'}(h))>0$. By the induction hypothesis, $\operatorname{up}_{\pi'}(h) \in Q$. Since $(Q,\mathcal D,B,U)$ is well-formed, the closure condition for one-step updates applies to $\operatorname{up}_{\pi'}(h) \in Q$ and yields $\operatorname{up}_{D,j}(\operatorname{up}_{\pi'}(h)) \in Q$. By the recursive clause, $\operatorname{up}_{\pi';(D,j)}(h)=\operatorname{up}_{D,j}(\operatorname{up}_{\pi'}(h))$, so $\operatorname{up}_{\pi';(D,j)}(h) \in Q$.
\end{proof}

\subsection{History probabilities from posterior quantities}\label{subsec:history-probabilities}

\begin{definition}[Posterior quantities]\label{def:posterior-quantities}
Let $h \in Q$ be a prior and let $\pi$ be a history. For $i \in I$, the \emph{posterior coordinate} is $h_i^{\pi}(h):=(\operatorname{up}_{\pi}(h))_i$ if $\operatorname{up}_{\pi}(h)$ is defined, and $h_i^{\pi}(h):=0$ otherwise.

For $D \in \mathcal{D}$ and $j \in O_D$, the \emph{posterior outcome probability} is $q_j^{D,\pi}(h):=q_j^D(\operatorname{up}_{\pi}(h))$ if $\operatorname{up}_{\pi}(h)$ is defined, and $q_j^{D,\pi}(h):=0$ otherwise.
\end{definition}

The value $0$ in the undefined case is a convention that renders the term semantics total. The terms indexed by histories can therefore occur in formulas without an executability guard. When a formula is intended to express a posterior commitment on an executable branch, these terms may be combined with a suitable guard. The convention keeps all terms real-valued at every state while preserving the partial character of the update clauses.

\begin{definition}[History probability]\label{def:history-probability}
The \emph{history probability} $r_\pi(h)$ of a history $\pi$ at a prior $h \in Q$ is defined by recursion as follows.
\begin{itemize}
\item $r_{\varepsilon}(h):=1$.
\item $r_{\pi;(D,j)}(h):=r_{\pi}(h)\,q_j^{D,\pi}(h)$.
\end{itemize}
\end{definition}

An induction on $\pi$ shows that $r_\pi(h)\ge0$ for every $h\in Q$. At an executable prefix, the next factor is a Born probability evaluated at a prior in $Q$. At a nonexecutable prefix, it is $0$ by convention.

When $\operatorname{up}_{\pi}(h)$ is defined, the recursive clause for an extended history can be written as $r_{\pi;(D,j)}(h)=r_{\pi}(h)\,q_j^D(\operatorname{up}_{\pi}(h))$. In that case, the probability of the extended history is the probability of the preceding history times the conditional probability of the next recorded outcome, evaluated at the posterior reference prior after the preceding history.

The numerical history probability is total, whereas the recursive update is partial. Thus $r_{\pi;(D,j)}(h)$ is always defined as a real number, even when $\operatorname{up}_{\pi;(D,j)}(h)$ is not defined. The relation between zero history probability and nonexecutability is made precise in the next subsection.

\subsection{Executability criterion}\label{subsec:history-probability-criterion}

The following lemma relates the partial recursion for posterior reference priors to the total recursion for history probabilities. A history is executable precisely when the partial update is defined. The lemma shows that this condition is equivalent to positive history probability.

\begin{lemma}[Characterization of executability]\label{lem:definability}
For every admissible prior $h \in Q$ and every history $\pi$, the posterior reference prior $\operatorname{up}_{\pi}(h)$ is defined if and only if $r_{\pi}(h)>0$. Consequently, for every $D \in \mathcal{D}$ and every $j \in O_D$, the posterior reference prior $\operatorname{up}_{\pi;(D,j)}(h)$ is defined if and only if $r_{\pi;(D,j)}(h)>0$, and $r_{\pi;(D,j)}(h)>0$ if and only if both $r_{\pi}(h)>0$ and $q_j^{D,\pi}(h)>0$.
\end{lemma}

\begin{proof}
The proof is by induction on the length of $\pi$.

Case $\pi \equiv \varepsilon$. The posterior reference prior $\operatorname{up}_{\varepsilon}(h)=h$ is always defined and $r_{\varepsilon}(h)=1>0$, so the claim holds.

Case $\pi \equiv \pi';(D,j)$. Assume the claim holds for $\pi'$. By definition, $\operatorname{up}_{\pi';(D,j)}(h)$ is defined if and only if $\operatorname{up}_{\pi'}(h)$ is defined and $q_j^D(\operatorname{up}_{\pi'}(h))>0$. By the induction hypothesis, this is equivalent to $r_{\pi'}(h)>0$ and, because $\operatorname{up}_{\pi'}(h)$ is then defined, to $r_{\pi'}(h)>0$ and $q_j^{D,\pi'}(h)=q_j^D(\operatorname{up}_{\pi'}(h))>0$. Since $r_{\pi';(D,j)}(h)=r_{\pi'}(h)\,q_j^{D,\pi'}(h)$ and both factors are nonnegative, this is equivalent to $r_{\pi';(D,j)}(h)>0$.

This proves the first statement for $\pi';(D,j)$. The final equivalence stated in the lemma is the condition just obtained.
\end{proof}

Combining Lemmas~\ref{lem:history-closure} and~\ref{lem:definability}, we obtain that whenever $r_{\pi}(h)>0$, the posterior reference prior $\operatorname{up}_{\pi}(h)$ belongs to $Q$. Lemma~\ref{lem:definability} also allows executability and posterior evaluation to be expressed uniformly in terms of $r_{\pi}(h)$ and $q_j^{D,\pi}(h)$.

\subsection{Syntax}\label{subsec:syntax}

The \emph{object language} describes reference priors, actual outcome probabilities, history probabilities, posterior coordinates, and guarded dynamic evolution along histories.

The language is built over a fixed well-formed core datum. In particular, the measurements in $\mathcal D$, their outcome sets $O_D$, and the reference index set $I$ determine the indexed symbols of the language, while the fixed kernel families $B$ and $U$ determine their semantic interpretation. In addition, the language contains constant symbols used in numerical comparisons. It includes symbols for all rational numbers, interpreted by their usual values, and may include additional symbols naming real numbers. The permitted additional constants differ between the unrestricted semantic setting and the effective setting. See Subsection~\ref{subsec:semantic-settings}.

\begin{definition}[Object language]\label{def:object-language}
Terms are generated by the grammar $t ::= a \mid h_i \mid q_j^D \mid r_\pi \mid h_i^\pi \mid q_j^{D,\pi} \mid (t+t) \mid (-t) \mid (t \cdot t)$, where $a$ ranges over the chosen constant symbols, $i \in I$, $D \in \mathcal{D}$, $j \in O_D$, and $\pi$ ranges over histories.

Formulas are generated by the grammar $\varphi ::= \mathsf{Last}(D,j) \mid (t \ge 0) \mid \neg \varphi \mid (\varphi \wedge \psi) \mid [\pi]\varphi$. The remaining notation consists of abbreviations. For terms, subtraction, equality, and strict inequality are introduced by $(t-s):=(t+(-s))$, $(t=s):=(((t-s) \ge 0) \wedge ((s-t) \ge 0))$, and $(t>0):=\neg((-t) \ge 0)$. For Boolean notation, we use $\varphi \vee \psi := \neg(\neg \varphi \wedge \neg \psi)$, $\varphi \to \psi := (\neg \varphi \vee \psi)$, $\varphi \leftrightarrow \psi := ((\varphi \to \psi) \wedge (\psi \to \varphi))$, and $\top := (0 \ge 0)$. The dual possibility operator is $\langle \pi \rangle \varphi := \neg [\pi]\neg \varphi$. Finally, the one-step box notation is $[D:j]\varphi := [\varepsilon;(D,j)]\varphi$.
\end{definition}

The basic terms describe current and posterior probabilistic quantities. For a current reference prior, $h_i$ is its $i$th coordinate, $q_j^D$ the probability assigned to outcome $j$ of measurement $D$, and $r_\pi$ the probability of history $\pi$. The terms $h_i^\pi$ and $q_j^{D,\pi}$, which are indexed by histories, express the corresponding posterior quantities after $\pi$.

We use $\mathcal L$ for the \emph{full language}, including the dynamic operators $[\pi]$, and $\mathcal L_0$ for the \emph{fragment without dynamic operators}, obtained by omitting the clause $[\pi]\varphi$ from the formation rules while retaining the terms indexed by histories. The dynamic operators provide a direct formulation of the guarded dynamic semantics, while the global reduction theorem in Section~\ref{sec:metatheory} shows that every formula of $\mathcal L$ is equivalent to a formula in $\mathcal L_0$.

\subsection{Truth conditions for logical states}\label{subsec:truth-conditions}

\begin{definition}[Logical states]\label{def:logical-states}
A \emph{logical state} is a pair $\sigma=(h,l)$, where $h \in Q$ is an admissible prior and $l$ is a \emph{record marker}. The marker is either $\bot$ or a pair $(D,j)$ with $D \in \mathcal{D}$ and $j \in O_D$.
\end{definition}

The component $h$ carries the current normative probabilistic commitment in reference coordinates. The component $l$ carries the most recent recorded outcome, with $\bot$ marking its absence.

\begin{definition}[Satisfaction relation]\label{def:satisfaction-relation}
The denotation of terms at a logical state $\sigma=(h,l)$ is defined recursively from the basic semantic quantities introduced in Subsections~\ref{subsec:preliminaries-consistency-bounds}, \ref{subsec:semantic-data-kernels}, and \ref{subsec:history-probabilities} as follows.
\begin{itemize}
\item $\llbracket h_i \rrbracket_\sigma := h_i$.
\item $\llbracket q_j^D \rrbracket_\sigma := q_j^D(h)$.
\item $\llbracket r_\pi \rrbracket_\sigma := r_\pi(h)$.
\item $\llbracket h_i^\pi \rrbracket_\sigma := h_i^\pi(h)$.
\item $\llbracket q_j^{D,\pi} \rrbracket_\sigma := q_j^{D,\pi}(h)$.
\item $\llbracket a \rrbracket_\sigma := a$, where the occurrence of $a$ on the right denotes the real number named by the constant symbol.
\item $\llbracket t+u \rrbracket_\sigma := \llbracket t \rrbracket_\sigma+\llbracket u \rrbracket_\sigma$.
\item $\llbracket -t \rrbracket_\sigma := -\llbracket t \rrbracket_\sigma$.
\item $\llbracket t\cdot u \rrbracket_\sigma := \llbracket t \rrbracket_\sigma\cdot\llbracket u \rrbracket_\sigma$.
\end{itemize}

Truth of formulas is defined by recursion on formula complexity as follows.
\begin{itemize}
\item $\sigma \models \mathsf{Last}(D,j)$ if and only if $l=(D,j)$.
\item $\sigma \models (t \ge 0)$ if and only if $\llbracket t \rrbracket_\sigma \ge 0$.
\item $\sigma \models \neg\varphi$ if and only if $\sigma \not\models \varphi$.
\item $\sigma \models \varphi \wedge \psi$ if and only if $\sigma \models \varphi$ and $\sigma \models \psi$.
\item $\sigma \models [\varepsilon]\varphi$ if and only if $\sigma \models \varphi$.
\item $\sigma \models [\pi]\varphi$ if and only if either $r_\pi(h)=0$ or $(\operatorname{up}_{\pi}(h),\operatorname{last}(\pi)) \models \varphi$, for nonempty $\pi$.
\end{itemize}
\end{definition}

Let $\mathrm{Rec}:=\{\bot\} \cup \{(D,j) : D \in \mathcal{D},\ j \in O_D\}$ be the finite set of record markers.

\begin{definition}[Validity over a well-formed core datum]\label{def:validity-over-well-formed-core-datum}
A formula $\varphi \in \mathcal L$ is \emph{valid over} $(Q,\mathcal D,B,U)$, written $\models_{(Q,\mathcal D,B,U)}\varphi$, if $(h,l) \models \varphi$ for every $h \in Q$ and every $l \in \mathrm{Rec}$.
\end{definition}

\medskip

\begin{remark}[Reachability of logical states]\label{rem:logical-states-reachability}
Validity is taken over every logical state $(h,l)$ with $h \in Q$ and $l \in \mathrm{Rec}$, independently of reachability from an initial marker by executable histories. This choice keeps the semantics uniform and makes validity independent of a chosen initial record. Formulas containing $\mathsf{Last}$ are therefore tested at every possible record marker.
\end{remark}

This is the guarded partial correctness interpretation familiar from dynamic logic. The guard $r_\pi(h)=0$ covers nonexecutable histories, while Lemma~\ref{lem:definability} ensures that the posterior state is well-defined whenever $r_\pi(h)>0$.

The term language contains both current quantities and posterior quantities indexed by histories, while truth is evaluated at a single logical state. Here the dynamic operators are evaluated at recursively generated posterior states, so posterior commitment depends on the current state through guarded executability.

\subsection{Examples of guarded QBist formulas}\label{subsec:guarded-examples}

The preceding clauses define the semantics abstractly. The following examples show the kind of QBist assertions expressed by the language before any Hilbert space realization is imposed. Their truth values depend on the chosen well-formed core datum and the current admissible prior. Throughout these examples, let $\sigma=(h,l)$ be a logical state. We write $\langle D:j\rangle\varphi$ as an abbreviation for $\langle \varepsilon;(D,j)\rangle\varphi$.

\medskip

\begin{example}[Posterior commitment on an executable branch]\label{ex:guarded-executability-posterior}
Let $D,E\in\mathcal D$, $j\in O_D$, $k\in O_E$, and let $c$ be a constant symbol of the language.

The formula $\langle D:j\rangle\top$ says exactly that the actual outcome $j$ of $D$ is executable at the current prior. One has $\sigma\models \langle D:j\rangle\top$ if and only if $q_j^D(h)>0$. Thus the Born probability attached to the current reference prior determines the truth of the corresponding possibility formula.

The formula $\langle D:j\rangle(q_k^E-c\ge 0)$ is stronger. It requires the outcome $(D,j)$ to be executable and the posterior probability assigned to outcome $k$ of the further actual measurement $E$ to be at least $c$ after the corresponding update. Equivalently, $q_j^D(h)>0$ and $q_k^E(\operatorname{up}_{D,j}(h))\ge c$. This formula therefore states a posterior probability constraint on an executable branch.

The boxed formula $[D:j](q_k^E-c\ge0)$ expresses a different condition. It says that either $q_j^D(h)=0$, so the branch is nonexecutable, or else the posterior inequality $q_k^E(\operatorname{up}_{D,j}(h))\ge c$ holds. Thus $[D:j]\varphi$ expresses a conditional normative assessment, whereas $\langle D:j\rangle\varphi$ also requires the branch to be executable.
\end{example}

\medskip

\begin{example}[Branches with zero probability]\label{ex:guarded-zero-branches}
Let $D\in\mathcal D$ and $j\in O_D$. For every formula $\varphi$, both $(q_j^D=0)\to [D:j]\varphi$ and $(q_j^D=0)\to \neg\langle D:j\rangle\top$ are valid over every well-formed core datum. Thus, when the current Born probability of a branch is zero, every corresponding boxed formula is true vacuously and $\langle D:j\rangle\top$ is false. A nonvacuous posterior assertion can be written as $\langle D:j\rangle\top\wedge [D:j]\varphi$, or as $\langle D:j\rangle\varphi$ when the intended assertion is $\varphi$ at the posterior state.
\end{example}

\medskip

\begin{example}[Sequential availability]\label{ex:guarded-sequential-availability}
Let $D,E\in\mathcal D$, $j\in O_D$, and $k\in O_E$, and put $\pi:=\varepsilon;(D,j);(E,k)$. The formula $\langle D:j\rangle\langle E:k\rangle\top$ says that the two-step history $\pi$ is executable at the current prior. By Lemma~\ref{lem:definability}, this is equivalent at $\sigma$ to $r_\pi(h)>0$, and hence to the conjunction of $q_j^D(h)>0$ and $q_k^E(\operatorname{up}_{D,j}(h))>0$.

By contrast, $[D:j]\langle E:k\rangle\top$ says that either the first branch $(D,j)$ is nonexecutable, or, if it is executable, the further outcome $(E,k)$ remains available after the corresponding update. This formula is true at $\sigma$ if either $q_j^D(h)=0$ or $q_k^{E,\varepsilon;(D,j)}(h)>0$. It therefore expresses a stability condition on possible future experiences after a specified actual experience, while preserving the vacuous truth behavior of the guarded box on nonexecutable first branches.

The same formulas also illustrate the global reduction theorem of Subsection~\ref{subsec:global-reduction}. Up to the abbreviations introduced above, $[D:j]\langle E:k\rangle\top$ is represented in $\mathcal L_0$ by $(q_j^D=0)\vee(q_k^{E,\varepsilon;(D,j)}>0)$, while $\langle D:j\rangle\langle E:k\rangle\top$ is represented by $(r_{\varepsilon;(D,j);(E,k)}>0)$. Although the dynamic formulas are reducible to numerical conditions in $\mathcal L_0$, they display the sequence of updates directly.
\end{example}

\medskip

\begin{example}[Dependence of posterior commitments on order]\label{ex:guarded-order-dependence}
Let $D,E\in\mathcal D$, let $j\in O_D$ and $k\in O_E$, and define $\pi_0:=\varepsilon;(D,j);(E,k)$ and $\pi_1:=\varepsilon;(E,k);(D,j)$. The formula $\langle \pi_0\rangle\top\wedge \langle \pi_1\rangle\top\wedge (q_j^{D,\pi_1}-q_j^{D,\pi_0}>0)$ says that both histories are executable and that the posterior probability assigned to outcome $j$ in a subsequent performance of $D$ is greater after $\pi_1$ than after $\pi_0$.

The language evaluates claims of this kind relative to a chosen well-formed core datum and a chosen current prior. It compares two posterior commitments generated from the same current reference prior by two different histories. The formulas of the guarded language provide comparisons between posterior commitments produced by alternative actual histories, while a static description of the admissible prior space $Q$ specifies the allowed priors.
\end{example}

\medskip

\begin{example}[Comparison of reference coordinates with actual outcome probabilities]\label{ex:guarded-reference-actual-comparison}
The language can also express mixed assertions involving reference coordinates and actual outcome probabilities. For a history $\pi$, an index $i\in I$, a measurement $D\in\mathcal D$, and an outcome $j\in O_D$, the formula $\langle \pi\rangle\top\wedge (h_i^\pi-h_i>0)\wedge (q_j^D-q_j^{D,\pi}>0)$ says that the history is executable, that the posterior probability assigned to reference outcome $i$ is larger than its current probability, and that the posterior probability of the actual outcome $(D,j)$ is smaller than its current probability.

Such a formula is useful because the reference side and the actual side play different roles. The term $h_i^\pi$ concerns the posterior coordinate relative to the fixed reference measurement, while $q_j^{D,\pi}$ concerns the posterior probability of a possible actual outcome. The formula therefore compares the change in a reference coordinate with the change in an actual outcome probability.
\end{example}

\medskip

\begin{example}[Record markers in guarded updates]\label{ex:guarded-record-marker}
Let $D,E\in\mathcal D$, $j\in O_D$, $k\in O_E$, and let $c$ be a constant symbol of the language. The formula $[D:j]\mathsf{Last}(D,j)$ is valid over every well-formed core datum and displays the record component of the semantics. If the branch $(D,j)$ is executable, the updated logical state has record marker $(D,j)$. If the branch is not executable, the box is true vacuously.

A formula such as $\langle D:j\rangle\bigl(\mathsf{Last}(D,j)\wedge(q_k^E-c\ge0)\bigr)$ combines the updated marker with a posterior numerical assertion. It says that the outcome $(D,j)$ is executable, that the resulting record marker is $(D,j)$, and that the posterior probability of the later outcome $(E,k)$ is at least $c$. The $\mathsf{Last}(D,j)$ conjunct makes the updated marker explicit alongside the numerical condition.
\end{example}

\subsection{Two semantic settings}\label{subsec:semantic-settings}

The abstract language can be interpreted in two settings, depending on the role assigned to the theory.

In the \emph{unrestricted semantic setting}, the language may contain constant symbols naming arbitrary real numbers, each interpreted by its designated value. The background arithmetic is the ordered real field expanded by these constants. This setting is appropriate for presenting the semantics and proving validity statements over real-valued models. It is also used for the quantum realization in Section~\ref{sec:quantum-realization}, since the structure constants arising there need not be given effectively.

In the \emph{effective setting}, one works with an \emph{effectively presented real closed field} $K \subseteq \mathbb{R}$ containing all structure constants of the chosen well-formed core datum. The object language is restricted to constant symbols naming elements of $K$, and each such symbol is supplied with the code of its designated element in the fixed effective presentation of $K$. The background arithmetic is the first-order theory of real closed fields in the ordered ring signature expanded by constant symbols for elements of $K$, as in the standard model-theoretic treatment of real closed fields \cite{Marker2002}.

The two settings use the same formation rules, the same intended semantics of histories and updates, and the same guarded dynamic interpretation. They differ in the metatheoretic resources available for studying validity. The effective results of Section~\ref{sec:metatheory} are proved for effectively semialgebraic data in the effective setting.

\section{Metatheory}\label{sec:metatheory}

The global reduction theorem requires only well-formedness. Throughout Subsections~\ref{subsec:global-reduction} and~\ref{subsec:reduction-correctness}, fix an arbitrary well-formed core datum $(Q,\mathcal D,B,U)$. The effective assumptions needed for the first-order translation are introduced afterward.

\subsection{Global reduction}\label{subsec:global-reduction}

We now prove the global reduction theorem for the guarded dynamic language. Recall that $\mathcal L$ is the full language, whereas $\mathcal L_0$ is its fragment without dynamic operators, obtained by omitting the clause $[\pi]\varphi$ from the formation rules while retaining the terms indexed by histories.

The reduction proof uses two preliminary constructions, posterior states and concatenation of histories.

\begin{definition}[Posterior state]\label{def:posterior-state}
Let $\sigma=(h,l)$ be a logical state. The \emph{posterior state} after a history $\pi$ is denoted by $\sigma_\pi$. Set $\sigma_\varepsilon := \sigma$. If $\pi$ is nonempty and $r_\pi(h)>0$, set $\sigma_\pi := (\operatorname{up}_{\pi}(h),\operatorname{last}(\pi))$.
\end{definition}

The state $\sigma_\pi$ is defined if and only if either $\pi \equiv \varepsilon$ or the history $\pi$ is executable at $\sigma$.

To analyze nested dynamic operators, we also need an operation that appends one history to another. Since histories are built by appending one outcome at a time, concatenation follows the same recursive pattern.

\begin{definition}[Concatenation of histories]\label{def:history-concatenation}
If $\pi_0$ and $\pi_1$ are histories, their \emph{concatenation} $\pi_0;\pi_1$ is defined by recursion on $\pi_1$ as follows.
\begin{itemize}
\item $\pi_0;\varepsilon := \pi_0$.
\item $\pi_0;(\pi_1;(D,j)) := (\pi_0;\pi_1);(D,j)$.
\end{itemize}
\end{definition}

The first lemma describes the behavior of probabilities and posterior reference priors under concatenation. It is the formal backbone of the reduction argument, because nested dynamic operators will eventually be analyzed by splitting a long history into an initial segment and a remainder.

\begin{lemma}[Concatenation identities]\label{lem:concatenation-identities}
Let $h \in Q$, and let $\pi_0$ and $\pi_1$ be histories.
\begin{enumerate}
\item If $\operatorname{up}_{\pi_0}(h)$ is defined, then $r_{\pi_0;\pi_1}(h)=r_{\pi_0}(h)\,r_{\pi_1}(\operatorname{up}_{\pi_0}(h))$.
\item If $\operatorname{up}_{\pi_0}(h)$ is defined, then $\operatorname{up}_{\pi_0;\pi_1}(h)$ is defined if and only if $\operatorname{up}_{\pi_1}(\operatorname{up}_{\pi_0}(h))$ is defined. Whenever these are defined, they are equal.
\end{enumerate}
\end{lemma}

\begin{proof}
We prove both statements simultaneously by induction on the length of $\pi_1$.

Base case $\pi_1 \equiv \varepsilon$. One has $\pi_0;\varepsilon=\pi_0$, $r_\varepsilon(\operatorname{up}_{\pi_0}(h))=1$, and $\operatorname{up}_{\varepsilon}(\operatorname{up}_{\pi_0}(h))=\operatorname{up}_{\pi_0}(h)$. Therefore $r_{\pi_0;\varepsilon}(h)=r_{\pi_0}(h)=r_{\pi_0}(h)\,r_\varepsilon(\operatorname{up}_{\pi_0}(h))$, which proves part (i), and $\operatorname{up}_{\pi_0;\varepsilon}(h)=\operatorname{up}_{\pi_0}(h)=\operatorname{up}_{\varepsilon}(\operatorname{up}_{\pi_0}(h))$, which proves part (ii).

Induction step. Assume that $\pi_1$ is nonempty, so $\pi_1 \equiv \pi';(D,j)$, and that both statements hold for $\pi'$.

We first prove part (i) for $\pi';(D,j)$. By definition of concatenation and the recursive clause for history probabilities, $r_{\pi_0;(\pi';(D,j))}(h)=r_{(\pi_0;\pi');(D,j)}(h)=r_{\pi_0;\pi'}(h)\,q_j^{D,\pi_0;\pi'}(h)$. By the induction hypothesis for part (i), $r_{\pi_0;\pi'}(h)=r_{\pi_0}(h)\,r_{\pi'}(\operatorname{up}_{\pi_0}(h))$. Hence $r_{\pi_0;(\pi';(D,j))}(h)=r_{\pi_0}(h)\,r_{\pi'}(\operatorname{up}_{\pi_0}(h))\,q_j^{D,\pi_0;\pi'}(h)$.

Subcase $r_{\pi'}(\operatorname{up}_{\pi_0}(h))=0$. The right hand side is $0$. On the other hand, by the induction hypothesis for part (i), $r_{\pi_0;\pi'}(h)=0$, so the recursive definition also gives $r_{\pi_0;(\pi';(D,j))}(h)=0$. Since $r_{\pi';(D,j)}(\operatorname{up}_{\pi_0}(h))=r_{\pi'}(\operatorname{up}_{\pi_0}(h))\,q_j^{D,\pi'}(\operatorname{up}_{\pi_0}(h))=0$, we conclude that $r_{\pi_0;(\pi';(D,j))}(h)=r_{\pi_0}(h)\,r_{\pi';(D,j)}(\operatorname{up}_{\pi_0}(h))$.

Subcase $r_{\pi'}(\operatorname{up}_{\pi_0}(h))>0$. By Lemma~\ref{lem:definability}, the posterior reference prior $\operatorname{up}_{\pi'}(\operatorname{up}_{\pi_0}(h))$ is defined. Since $\operatorname{up}_{\pi_0}(h)$ is defined by assumption, the induction hypothesis for part (ii) yields that $\operatorname{up}_{\pi_0;\pi'}(h)$ is defined and $\operatorname{up}_{\pi_0;\pi'}(h)=\operatorname{up}_{\pi'}(\operatorname{up}_{\pi_0}(h))$. Therefore $q_j^{D,\pi_0;\pi'}(h)=q_j^D(\operatorname{up}_{\pi_0;\pi'}(h))=q_j^D(\operatorname{up}_{\pi'}(\operatorname{up}_{\pi_0}(h)))=q_j^{D,\pi'}(\operatorname{up}_{\pi_0}(h))$. Hence
$r_{\pi_0;(\pi';(D,j))}(h)=r_{\pi_0}(h)\,r_{\pi'}(\operatorname{up}_{\pi_0}(h))\,q_j^{D,\pi'}(\operatorname{up}_{\pi_0}(h))=r_{\pi_0}(h)\,r_{\pi';(D,j)}(\operatorname{up}_{\pi_0}(h))$.
This proves part (i) for $\pi';(D,j)$.

We next prove part (ii) for $\pi';(D,j)$. By definition of recursive updates, $\operatorname{up}_{\pi_0;(\pi';(D,j))}(h)$ is defined if and only if both $\operatorname{up}_{\pi_0;\pi'}(h)$ is defined and $q_j^D(\operatorname{up}_{\pi_0;\pi'}(h))>0$. By the induction hypothesis for part (ii), this is equivalent to $\operatorname{up}_{\pi'}(\operatorname{up}_{\pi_0}(h))$ being defined and $q_j^D(\operatorname{up}_{\pi'}(\operatorname{up}_{\pi_0}(h)))>0$. Applying the recursive definition of updates once more, this is equivalent to $\operatorname{up}_{\pi';(D,j)}(\operatorname{up}_{\pi_0}(h))$ being defined.

Assume these equivalent conditions hold. Then, by the recursive clause and the induction hypothesis for part (ii), $\operatorname{up}_{\pi_0;(\pi';(D,j))}(h)=\operatorname{up}_{D,j}(\operatorname{up}_{\pi_0;\pi'}(h))=\operatorname{up}_{D,j}(\operatorname{up}_{\pi'}(\operatorname{up}_{\pi_0}(h)))=\operatorname{up}_{\pi';(D,j)}(\operatorname{up}_{\pi_0}(h))$. Thus part (ii) also holds for $\pi';(D,j)$.

By induction, both statements hold for every history $\pi_1$.
\end{proof}

The following definition collects the operations used in the reduction. When a formula is evaluated after a history $\pi$, the term translation replaces each current quantity by the corresponding posterior quantity, and every further dynamic operator is rewritten as a static executability condition together with evaluation after the concatenated history.

\begin{definition}[Reduction operations]\label{def:reduction-operations}
The reduction operations are specified as follows.

\begin{enumerate}
\item \emph{Auxiliary terms.}
Let $\pi$ be a nonempty history. For each history $\pi_1$, the auxiliary term $R_{\pi_1}^\pi$ is defined by recursion on $\pi_1$ as follows.
\begin{itemize}
\item $R_\varepsilon^\pi := 1$.
\item $R_{\pi_1;(D,j)}^\pi := R_{\pi_1}^\pi \cdot q_j^{D,\pi;\pi_1}$.
\end{itemize}
\item \emph{Term translation.}
For each history $\pi$, the term translation associated with $\pi$ is the map $t \mapsto T_\pi(t)$. For the empty history, set $T_\varepsilon(t) := t$ for every term $t$. For a nonempty history $\pi$, define $T_\pi(t)$ by recursion on the structure of $t$ as follows.
\begin{itemize}
\item $T_\pi(a) := a$.
\item $T_\pi(h_i) := h_i^\pi$.
\item $T_\pi(q_j^D) := q_j^{D,\pi}$.
\item $T_\pi(r_{\pi_1}) := R_{\pi_1}^\pi$.
\item $T_\pi(h_i^{\pi_1}) := h_i^{\pi;\pi_1}$.
\item $T_\pi(q_j^{D,\pi_1}) := q_j^{D,\pi;\pi_1}$.
\item $T_\pi(u+v) := T_\pi(u)+T_\pi(v)$.
\item $T_\pi(-u) := -T_\pi(u)$.
\item $T_\pi(u\cdot v) := T_\pi(u)\cdot T_\pi(v)$.
\end{itemize}

\item \emph{Formula translation.}
For each history $\pi$, the formula translation associated with $\pi$ is the map $\varphi \mapsto \operatorname{Red}_\pi(\varphi)$, defined by recursion on formula complexity as follows.
\begin{itemize}
\item $\operatorname{Red}_\pi(t \ge 0) := (T_\pi(t) \ge 0)$.
\item $\operatorname{Red}_\varepsilon(\mathsf{Last}(D,j)) := \mathsf{Last}(D,j)$.
\item $\operatorname{Red}_\pi(\mathsf{Last}(D,j)) := \top$, for nonempty $\pi$ with $\operatorname{last}(\pi)=(D,j)$.
\item $\operatorname{Red}_\pi(\mathsf{Last}(D,j)) := \neg\top$, for nonempty $\pi$ with $\operatorname{last}(\pi)\ne(D,j)$.
\item $\operatorname{Red}_\pi(\neg\psi) := \neg\operatorname{Red}_\pi(\psi)$.
\item $\operatorname{Red}_\pi(\psi \wedge \theta) := \operatorname{Red}_\pi(\psi) \wedge \operatorname{Red}_\pi(\theta)$.
\item $\operatorname{Red}_\pi([\varepsilon]\psi) := \operatorname{Red}_\pi(\psi)$.
\item $\operatorname{Red}_\pi([\pi_1]\psi) := \bigl(T_\pi(r_{\pi_1})=0\bigr) \vee \operatorname{Red}_{\pi;\pi_1}(\psi)$, for nonempty $\pi_1$.
\end{itemize}

\item \emph{Global reduction.}
The global reduction of a formula $\varphi$ is
$\operatorname{Red}(\varphi) := \operatorname{Red}_\varepsilon(\varphi)$.
\end{enumerate}
\end{definition}

For a nonempty history $\pi$, the auxiliary term $R_{\pi_1}^\pi$ represents the probability of $\pi_1$ evaluated at the posterior state obtained after $\pi$. Example~\ref{ex:guarded-sequential-availability} illustrates these operations for a two-step history.

Before turning to semantics, we state the purely syntactic fact that the reduction eliminates all dynamic operators. This lets the later argument split cleanly into a syntactic elimination step and a semantic correctness step.

\begin{lemma}[Reduction to the fragment without dynamic operators]\label{lem:static-fragment-reduction}
For every history $\pi$ and every formula $\varphi \in \mathcal L$, the formula $\operatorname{Red}_\pi(\varphi)$ belongs to $\mathcal L_0$.
\end{lemma}

\begin{proof}
We argue by induction on the complexity of $\varphi$.

Case $\varphi \equiv (t \ge 0)$. By definition, $\operatorname{Red}_\pi(\varphi)=(T_\pi(t)\ge 0)$, which is a formula of $\mathcal L_0$ because $\mathcal L_0$ contains all arithmetic atoms and no dynamic operator appears.

Case $\varphi \equiv \mathsf{Last}(D,j)$. There are two subcases.

Subcase $\pi \equiv \varepsilon$. $\operatorname{Red}_\varepsilon(\mathsf{Last}(D,j))=\mathsf{Last}(D,j)$.

Subcase $\pi \not\equiv \varepsilon$. $\operatorname{Red}_\pi(\mathsf{Last}(D,j))$ is either $\top$ or $\neg\top$. In each subcase the result lies in $\mathcal L_0$.

Case $\varphi \equiv \neg\psi$. By definition, $\operatorname{Red}_\pi(\neg\psi)=\neg\operatorname{Red}_\pi(\psi)$. By the induction hypothesis, $\operatorname{Red}_\pi(\psi)$ belongs to $\mathcal L_0$. Since $\mathcal L_0$ is closed under Boolean connectives, $\neg\operatorname{Red}_\pi(\psi)$ also belongs to $\mathcal L_0$.

Case $\varphi \equiv \psi \wedge \theta$. Then $\operatorname{Red}_\pi(\psi \wedge \theta)=\operatorname{Red}_\pi(\psi)\wedge\operatorname{Red}_\pi(\theta)$. By the induction hypothesis, both conjuncts belong to $\mathcal L_0$, hence so does their conjunction.

Case $\varphi \equiv [\varepsilon]\psi$. Then $\operatorname{Red}_\pi([\varepsilon]\psi)=\operatorname{Red}_\pi(\psi)$, which belongs to $\mathcal L_0$ by the induction hypothesis.

Case $\varphi \equiv [\pi_1]\psi$, where $\pi_1$ is nonempty. Then $\operatorname{Red}_\pi([\pi_1]\psi)=\bigl(T_\pi(r_{\pi_1})=0\bigr)\vee\operatorname{Red}_{\pi;\pi_1}(\psi)$. The first disjunct is an arithmetic formula, hence a formula of $\mathcal L_0$. By the induction hypothesis, the second disjunct also belongs to $\mathcal L_0$. Since $\mathcal L_0$ is closed under Boolean connectives, the whole disjunction belongs to $\mathcal L_0$.

This completes the induction.
\end{proof}

\subsection{Reduction correctness}\label{subsec:reduction-correctness}

The correctness proof begins with the term translation. The proof for formulas then combines this result with the Boolean and guarded dynamic clauses.

\begin{lemma}[Correctness of the term translation]\label{lem:term-translation-correctness}
Let $\sigma$ be a logical state and let $\pi$ be a history such that $\sigma_\pi$ is defined. Then for every term $t$, $\llbracket T_\pi(t) \rrbracket_\sigma = \llbracket t \rrbracket_{\sigma_\pi}$.
\end{lemma}

\begin{proof}
We argue by induction on the formation of $t$.

If $\pi \equiv \varepsilon$, then $\sigma_\varepsilon=\sigma$ and $T_\varepsilon(t)=t$, so the claim is immediate. For the remainder of the proof, assume that $\pi$ is nonempty.

Case $t \equiv a$. The claim is immediate since $T_\pi(a)=a$.

Case $t \equiv h_i$. We have $\llbracket T_\pi(h_i) \rrbracket_\sigma=\llbracket h_i^\pi \rrbracket_\sigma=h_i^\pi(h)$. Since $\sigma_\pi$ is defined, Lemma~\ref{lem:definability} implies that $\operatorname{up}_{\pi}(h)$ is defined, and by definition of posterior coordinates, $h_i^\pi(h) = (\operatorname{up}_{\pi}(h))_i = \llbracket h_i \rrbracket_{\sigma_\pi}$.

Case $t \equiv q_j^D$. We similarly obtain $\llbracket T_\pi(q_j^D) \rrbracket_\sigma=\llbracket q_j^{D,\pi} \rrbracket_\sigma=q_j^{D,\pi}(h)=q_j^D(\operatorname{up}_{\pi}(h))=\llbracket q_j^D \rrbracket_{\sigma_\pi}$.

Case $t \equiv r_{\pi_1}$. We prove $\llbracket R_{\pi_1}^\pi \rrbracket_\sigma = r_{\pi_1}(\operatorname{up}_{\pi}(h))$ by induction on the length of $\pi_1$.

Base case $\pi_1 \equiv \varepsilon$. Then $\llbracket R_\varepsilon^\pi \rrbracket_\sigma = 1 = r_\varepsilon(\operatorname{up}_{\pi}(h))$.

Induction step. Assume that $\pi_1$ is nonempty, so $\pi_1 \equiv \pi';(D,j)$, and that the claim holds for $\pi'$. Then $\llbracket R_{\pi';(D,j)}^\pi \rrbracket_\sigma=\llbracket R_{\pi'}^\pi \cdot q_j^{D,\pi;\pi'} \rrbracket_\sigma=\llbracket R_{\pi'}^\pi \rrbracket_\sigma \cdot \llbracket q_j^{D,\pi;\pi'} \rrbracket_\sigma$. By the induction hypothesis, the first factor is $r_{\pi'}(\operatorname{up}_{\pi}(h))$.

Subcase $r_{\pi'}(\operatorname{up}_{\pi}(h))=0$. The product is $0$, which equals $r_{\pi';(D,j)}(\operatorname{up}_{\pi}(h))$.

Subcase $r_{\pi'}(\operatorname{up}_{\pi}(h))>0$. By Lemma~\ref{lem:definability}, applying $\pi'$ to $\operatorname{up}_{\pi}(h)$ yields a defined posterior reference prior. Lemma~\ref{lem:concatenation-identities} then gives $\operatorname{up}_{\pi;\pi'}(h)=\operatorname{up}_{\pi'}(\operatorname{up}_{\pi}(h))$. Hence $\llbracket q_j^{D,\pi;\pi'} \rrbracket_\sigma=q_j^D(\operatorname{up}_{\pi;\pi'}(h))$ and $q_j^D(\operatorname{up}_{\pi;\pi'}(h))=q_j^{D,\pi'}(\operatorname{up}_{\pi}(h))$. Therefore $\llbracket R_{\pi';(D,j)}^\pi \rrbracket_\sigma=r_{\pi'}(\operatorname{up}_{\pi}(h)) \cdot q_j^{D,\pi'}(\operatorname{up}_{\pi}(h))=r_{\pi';(D,j)}(\operatorname{up}_{\pi}(h))$. This proves the claim for $r_{\pi_1}$.

Case $t \equiv h_i^{\pi_1}$. We have $\llbracket T_\pi(h_i^{\pi_1}) \rrbracket_\sigma=\llbracket h_i^{\pi;\pi_1} \rrbracket_\sigma=h_i^{\pi;\pi_1}(h)$.

Subcase $r_{\pi_1}(\operatorname{up}_{\pi}(h))=0$. By part (i) of Lemma~\ref{lem:concatenation-identities}, one has $r_{\pi;\pi_1}(h)=0$, so $\operatorname{up}_{\pi;\pi_1}(h)$ is undefined by Lemma~\ref{lem:definability}. Hence $h_i^{\pi;\pi_1}(h)=0$ by definition. On the other hand, $\operatorname{up}_{\pi_1}(\operatorname{up}_{\pi}(h))$ is undefined, so $\llbracket h_i^{\pi_1} \rrbracket_{\sigma_\pi}=0$ as well.

Subcase $r_{\pi_1}(\operatorname{up}_{\pi}(h))>0$. Lemma~\ref{lem:concatenation-identities} yields $\operatorname{up}_{\pi;\pi_1}(h)=\operatorname{up}_{\pi_1}(\operatorname{up}_{\pi}(h))$, and therefore $h_i^{\pi;\pi_1}(h)=(\operatorname{up}_{\pi;\pi_1}(h))_i=(\operatorname{up}_{\pi_1}(\operatorname{up}_{\pi}(h)))_i=\llbracket h_i^{\pi_1} \rrbracket_{\sigma_\pi}$.

Case $t \equiv q_j^{D,\pi_1}$. The argument is analogous.

For $t \equiv u+v$, $t \equiv -u$, and $t \equiv u \cdot v$, the claim follows directly from the induction hypothesis and the semantic clauses for the arithmetic operations.
\end{proof}

Lemma~\ref{lem:term-translation-correctness} handles the atomic formulas. The Boolean cases follow from the induction hypothesis, and the guarded dynamic clause supplies the remaining case.

\begin{proposition}[Correctness of the reduction]\label{prop:reduction-correctness}
Let $\sigma$ be a logical state and let $\pi$ be a history such that $\sigma_\pi$ is defined. Then for every formula $\varphi \in \mathcal L$, $\sigma \models \operatorname{Red}_\pi(\varphi)$ if and only if $\sigma_\pi \models \varphi$.
\end{proposition}

\begin{proof}
We prove the statement by induction on the complexity of $\varphi$, simultaneously for all logical states $\sigma$ and all histories $\pi$ such that $\sigma_\pi$ is defined. Thus, in the dynamic case, the induction hypothesis may be applied with the longer history $\pi;\pi_1$ whenever the corresponding posterior state is defined.

Case $\varphi \equiv (t \ge 0)$. By the semantic clause for arithmetic atoms, $\sigma \models \operatorname{Red}_\pi(t \ge 0)$ if and only if $\llbracket T_\pi(t) \rrbracket_\sigma \ge 0$. Lemma~\ref{lem:term-translation-correctness} makes this condition equivalent to $\llbracket t \rrbracket_{\sigma_\pi} \ge 0$, which holds if and only if $\sigma_\pi \models (t \ge 0)$.

Case $\varphi \equiv \mathsf{Last}(D,j)$. There are two subcases.

Subcase $\pi \equiv \varepsilon$. $\sigma_\varepsilon=\sigma$ and $\operatorname{Red}_\varepsilon(\mathsf{Last}(D,j))=\mathsf{Last}(D,j)$, so the claim is immediate.

Subcase $\pi \not\equiv \varepsilon$. The record component of $\sigma_\pi$ is $\operatorname{last}(\pi)$ by definition. Hence $\sigma_\pi \models \mathsf{Last}(D,j)$ if and only if $\operatorname{last}(\pi)=(D,j)$, which is the condition under which $\operatorname{Red}_\pi(\mathsf{Last}(D,j))$ is $\top$. It is $\neg\top$ otherwise.

Cases $\varphi \equiv \neg\psi$ and $\varphi \equiv \psi \wedge \theta$. These cases are immediate from the induction hypothesis and the classical truth clauses.

It remains to consider the dynamic case $\varphi \equiv [\pi_1]\psi$.

Subcase $\pi_1 \equiv \varepsilon$. Then by definition, $\operatorname{Red}_\pi([\varepsilon]\psi)=\operatorname{Red}_\pi(\psi)$. By the induction hypothesis, $\sigma \models \operatorname{Red}_\pi(\psi)$ if and only if $\sigma_\pi \models \psi$. Since $[\varepsilon]\psi$ is equivalent by semantics to $\psi$, this is equivalent to $\sigma_\pi \models [\varepsilon]\psi$. For the remaining subcases, assume that $\pi_1$ is nonempty.

By definition, $\operatorname{Red}_\pi([\pi_1]\psi)=\bigl(T_\pi(r_{\pi_1})=0\bigr)\vee\operatorname{Red}_{\pi;\pi_1}(\psi)$. By Lemma~\ref{lem:term-translation-correctness}, $\sigma \models \bigl(T_\pi(r_{\pi_1})=0\bigr)$ if and only if $r_{\pi_1}(\operatorname{up}_{\pi}(h))=0$. By part (i) of Lemma~\ref{lem:concatenation-identities}, this is equivalent to $r_{\pi;\pi_1}(h)=0$ because $\sigma_\pi$ is defined and hence $r_\pi(h)>0$.

Subcase $r_{\pi_1}(\operatorname{up}_{\pi}(h))=0$. The first disjunct in $\operatorname{Red}_\pi([\pi_1]\psi)$ is true by the equivalence established above, so $\sigma \models \operatorname{Red}_\pi([\pi_1]\psi)$. The semantic clause for the dynamic operator also gives $\sigma_\pi \models [\pi_1]\psi$.

Subcase $r_{\pi_1}(\operatorname{up}_{\pi}(h))>0$. By Lemma~\ref{lem:definability}, the state obtained from $\sigma_\pi$ by executing $\pi_1$ is defined. Part (ii) of Lemma~\ref{lem:concatenation-identities} shows that its prior component is $\operatorname{up}_{\pi;\pi_1}(h)$. Since $\pi_1$ is nonempty, its record component is $\operatorname{last}(\pi_1)=\operatorname{last}(\pi;\pi_1)$. Thus this state is $\sigma_{\pi;\pi_1}$. In this subcase, the semantic clause gives $\sigma_\pi \models [\pi_1]\psi$ if and only if $\sigma_{\pi;\pi_1} \models \psi$. By the induction hypothesis, this is equivalent to $\sigma \models \operatorname{Red}_{\pi;\pi_1}(\psi)$. The first disjunct in $\operatorname{Red}_\pi([\pi_1]\psi)$ is false, so the last condition is equivalent to $\sigma \models \operatorname{Red}_\pi([\pi_1]\psi)$.

This completes the induction.
\end{proof}

The global reduction theorem is the special case for the empty history. It provides the bridge from $\mathcal L$ to $\mathcal L_0$.

\begin{theorem}[Global reduction theorem]\label{thm:global-reduction}
For every formula $\varphi \in \mathcal L$ and every logical state $\sigma$ over the fixed well-formed core datum $(Q,\mathcal D,B,U)$, $\sigma \models \varphi$ if and only if $\sigma \models \operatorname{Red}(\varphi)$. Consequently, $\models_{(Q,\mathcal D,B,U)} \varphi$ if and only if $\models_{(Q,\mathcal D,B,U)} \operatorname{Red}(\varphi)$.
\end{theorem}

\begin{proof}
Since $\sigma_\varepsilon=\sigma$, Proposition~\ref{prop:reduction-correctness} with $\pi \equiv \varepsilon$ gives $\sigma \models \operatorname{Red}_\varepsilon(\varphi)$ if and only if $\sigma \models \varphi$. By definition, $\operatorname{Red}(\varphi)=\operatorname{Red}_\varepsilon(\varphi)$. The second statement follows immediately by quantifying over all logical states.
\end{proof}

\subsection{Effectively semialgebraic data}\label{subsec:effectively-semialgebraic-data}

The first-order analysis applies to a class of well-formed core data for which validity can be treated by effective first-order methods. We use standard terminology from real algebraic geometry for semialgebraic presentations and effective quantifier elimination \cite{BasuPollackRoy2006}. The definitions below separate the ambient ordered field, the presentation of the admissible prior space, and the combined notion of an effective datum used in the subsequent translation.

\begin{definition}[Effectively presented real closed field]\label{def:effectively-presented-rcf}
An \emph{effectively presented real closed field} is a real closed ordered field $K$ whose underlying set is represented by a recursive subset of $\mathbb{N}$ and for which the field operations and the order relation are recursive on that coding.
\end{definition}

The field of real algebraic numbers with its standard effective presentation is a basic example. In what follows, let $K \subseteq \mathbb{R}$ be an effectively presented real closed field, and let $\Sigma_K$ be the ordered ring signature expanded by constant symbols for all elements of $K$. In the remainder of this section, a first-order formula means a formula over $\Sigma_K$, interpreted in $\mathbb{R}$ unless otherwise stated.

\begin{definition}[Effective semialgebraic presentation of the admissible prior space]\label{def:effective-semialgebraic-presentation}
Let $\bar{x}:=(x_i)_{i\in I}$. An \emph{effective semialgebraic presentation} of an admissible prior space $Q \subseteq \mathbb{R}^{d^2}$ over $K$ is a quantifier-free first-order formula $\Theta_Q(\bar{x})$ over $\Sigma_K$ such that $Q=\{h \in \mathbb{R}^{d^2} \mid \mathbb{R} \models \Theta_Q(h)\}$ and every $h \in \mathbb{R}^{d^2}$ satisfying $\Theta_Q$ belongs to $\Delta(I)$.
\end{definition}

\medskip

\begin{example}[Simplex presentation]\label{ex:effective-simplex-presentation}
The full probability simplex $\Delta(I)$ has an effective semialgebraic presentation over any effectively presented real closed field $K$. Such a presentation is a finite algebraic description of the admissible prior space by equations and inequalities over the chosen coefficient field. It is given by the quantifier-free formula $\Theta_\Delta(\bar{x}):=\bigl(\bigwedge_{i\in I}x_i\ge 0\bigr)\wedge\bigl(\sum_{i\in I}x_i=1\bigr)$. Both the coefficients occurring in this formula and the kernel values used in Example~\ref{ex:full-simplex-datum} are rational, so their codes can be computed in the fixed effective presentation of $K$. Together with Example~\ref{ex:full-simplex-datum}, this presentation therefore gives an effectively semialgebraic datum over $K$ when the constant symbols of the language name elements of $K$ and are supplied with their codes in the fixed effective presentation of $K$.
\end{example}

\begin{definition}[Effectively semialgebraic datum]\label{def:effective-semialgebraic-datum}
A well-formed core datum $(Q,\mathcal{D},B,U)$ is \emph{effectively semialgebraic} over $K$ if the following hold.
\begin{enumerate}
\item Every kernel value $B^D(j|i)$ and $U^{D,j}(m|i)$ belongs to $K$ and is supplied with a code in the fixed effective presentation of $K$.
\item Every constant symbol of the object language names an element of $K$ and is supplied with the code of that element in the fixed effective presentation of $K$.
\item The set $Q$ admits an effective semialgebraic presentation $\Theta_Q$ over $K$.
\end{enumerate}
\end{definition}

This restriction specifies the class for which validity in $\mathcal L_0$ can be analyzed by the standard proof theory of real closed fields.

For the remainder of Section~\ref{sec:metatheory}, fix an effectively semialgebraic well-formed core datum $(Q,\mathcal D,B,U)$ over $K$ and a defining formula $\Theta_Q$. Write $\mathfrak M := (Q,\mathcal D,B,U,\Theta_Q,K)$. The first-order translation and the subsequent metatheoretic arguments are relative to $\mathfrak M$.

The later decidability result also requires the reduction to be effectively computable.

\begin{proposition}[Effective constructibility of the reduction]\label{prop:reduction-effective}
In the effective setting, there is an effective procedure that, given a formula $\varphi \in \mathcal L$, computes the reduced formula $\operatorname{Red}(\varphi)$.
\end{proposition}

\begin{proof}
Histories are finite strings over the finite alphabet of outcome pairs $(D,j)$, and $\varepsilon$ is the empty string. Hence histories admit an effective coding, and concatenation $(\pi,\pi_1)\mapsto \pi;\pi_1$ is computable from that coding.

The map $(\pi,t)\mapsto T_\pi(t)$ is computed by primitive recursion on the formation of terms. The atomic clauses are explicit, and the algebraic clauses are obtained by effective composition. Using this term translation, the map $(\pi,\varphi)\mapsto\operatorname{Red}_\pi(\varphi)$ is computed by primitive recursion on the formation of formulas. The atomic and Boolean clauses are explicit, and the dynamic clause for $[\pi_1]\psi$ uses only the already computable operations $(\pi,\pi_1)\mapsto T_\pi(r_{\pi_1})$ and $(\pi,\pi_1)\mapsto\pi;\pi_1$.

Finally, $\operatorname{Red}(\varphi)$ is by definition $\operatorname{Red}_\varepsilon(\varphi)$. Hence the whole construction is effective.
\end{proof}

\subsection{First-order translation of the fragment without dynamic operators}\label{subsec:first-order-translation}

This subsection translates $\mathcal L_0$ into first-order formulas over real closed fields. Relative to admissible inputs in $Q$, a first-order formula $\Gamma(\bar{x},\bar{y})$ defines the graph of a possibly partial map $f$ when, for every $h\in Q$ and every real tuple $\bar{b}$ of the appropriate arity, $\mathbb{R}\models\Gamma(h,\bar{b})$ if and only if $f(h)$ is defined and $\bar{b}=f(h)$. We use such formulas for posterior reference priors, history probabilities, and term values, and then translate formulas of the object language into ordinary first-order formulas over the ordered field.

Let $\bar{x}=(x_i)_{i \in I}$, $\bar{u}=(u_i)_{i \in I}$, and $\bar{v}=(v_i)_{i \in I}$ be tuples of first-order variables. We write $\Theta_Q(\bar{x})$ for the quantifier-free formula defining $Q$ from Definition~\ref{def:effective-semialgebraic-presentation}.

\begin{definition}[Ordered ring terms for the effective datum]\label{def:ordered-ring-effective-terms}
The ordered ring terms used below are
\begin{itemize}
\item $\operatorname{Aff}_i(\bar{x}) := (d+1)x_i - 1/d$, for $i \in I$.
\item $\operatorname{Born}_{D,j}(\bar{u}) := \sum_{i \in I}\operatorname{Aff}_i(\bar{u})\,B^D(j|i)$, for $D \in \mathcal D$ and $j \in O_D$.
\item $\operatorname{Num}^{D,j}_m(\bar{u}) := \sum_{i \in I}\operatorname{Aff}_i(\bar{u})\,U^{D,j}(m|i)$, for $D \in \mathcal D$, $j \in O_D$, and $m \in I$.
\end{itemize}
\end{definition}

The terms $\operatorname{Aff}_i(\bar{x})$, $\operatorname{Born}_{D,j}(\bar{u})$, and $\operatorname{Num}^{D,j}_m(\bar{u})$ represent, respectively, the affine coordinate $\alpha_i$, the current Born probability $q_j^D(u)$, and the numerator in the $m$th coordinate of the update by $(D,j)$, where $\bar{u}$ consists of the coordinates of an admissible prior $u \in Q$. Consequently, whenever $\operatorname{Born}_{D,j}(\bar{u})>0$, the equality $\operatorname{up}_{D,j}(u)_m=\operatorname{Num}^{D,j}_m(\bar{u})/\operatorname{Born}_{D,j}(\bar{u})$ holds. These terms express the affine coefficients, current Born probabilities, and update numerators inside ordered ring formulas without taking division as a primitive operation.

Using these terms, we define a first-order formula for a one-step update.

\begin{definition}[First-order formula for a one-step update]\label{def:one-step-update-formula}
For $D \in \mathcal D$ and $j \in O_D$, set
$\operatorname{Step}_{D,j}(\bar{u},\bar{v}) :=
\Theta_Q(\bar{u})
\wedge (\operatorname{Born}_{D,j}(\bar{u})>0)
\wedge \Theta_Q(\bar{v})
\wedge \bigwedge_{m \in I}(v_m \cdot \operatorname{Born}_{D,j}(\bar{u})=\operatorname{Num}^{D,j}_m(\bar{u}))$.
\end{definition}

The formula $\operatorname{Step}_{D,j}(\bar{u},\bar{v})$ expresses that $\bar{u}$ is an admissible prior, that the outcome $(D,j)$ is executable at $\bar{u}$, and that $\bar{v}$ is the resulting admissible posterior reference prior. The coordinate equations express the normalized update without division.

\begin{definition}[First-order graph formulas for histories]\label{def:graph-formulas-for-histories}
For each history $\pi$, the formulas $\operatorname{Post}_\pi(\bar{x},\bar{u})$ and $\operatorname{Prob}_\pi(\bar{x},z)$ over $\Sigma_K$ are defined by recursion on $\pi$ as follows.

The formulas $\operatorname{Post}_\pi$ are given by
\begin{itemize}
\item $\operatorname{Post}_\varepsilon(\bar{x},\bar{u}) := \bigwedge_{i \in I}(u_i=x_i)$.
\item $\operatorname{Post}_{\pi;(D,j)}(\bar{x},\bar{v}) := \exists \bar{u}\bigl(\operatorname{Post}_\pi(\bar{x},\bar{u}) \wedge \operatorname{Step}_{D,j}(\bar{u},\bar{v})\bigr)$.
\end{itemize}

The formulas $\operatorname{Prob}_\pi$ are given by
\begin{itemize}
\item $\operatorname{Prob}_\varepsilon(\bar{x},z) := (z=1)$.
\item $\operatorname{Prob}_{\pi;(D,j)}(\bar{x},z') := \exists z\bigl(\operatorname{Prob}_\pi(\bar{x},z) \wedge \bigl((z=0 \wedge z'=0) \vee (z>0 \wedge \exists \bar{u}(\operatorname{Post}_\pi(\bar{x},\bar{u}) \wedge ((\operatorname{Born}_{D,j}(\bar{u})=0 \wedge z'=0) \vee (\operatorname{Born}_{D,j}(\bar{u})>0 \wedge z'=z \cdot \operatorname{Born}_{D,j}(\bar{u})))))\bigr)\bigr)$.
\end{itemize}
\end{definition}

These formulas encode the posterior and probabilistic effects of $\pi$ on admissible inputs. The formula $\operatorname{Post}_\pi(\bar{x},\bar{u})$ relates an admissible prior $\bar{x}$ to the posterior reference prior $\bar{u}$ obtained after $\pi$, while $\operatorname{Prob}_\pi(\bar{x},z)$ associates $z$ with the probability of $\pi$ at $\bar{x}$. In the recursive clause for $\operatorname{Prob}_{\pi;(D,j)}$, a zero probability for the preceding history or the next outcome forces the extended history probability to be $0$. When both quantities are positive, the extended history probability is the product of the probability of $\pi$ and the Born probability of the next outcome at the posterior reference prior after $\pi$.

\begin{proposition}[Effective constructibility of graph formulas for histories]\label{prop:graph-formulas-for-histories-effective}
There is an effective procedure that, given a history $\pi$, computes the formulas $\operatorname{Post}_\pi$ and $\operatorname{Prob}_\pi$.
\end{proposition}

\begin{proof}
The formulas are computed by primitive recursion on the length of $\pi$ from the explicit clauses of Definition~\ref{def:graph-formulas-for-histories}. All symbols and coefficients used in those clauses are effectively given.
\end{proof}

These graph formulas represent posterior reference priors and history probabilities on the admissible prior space.

\begin{lemma}[Correctness of graph formulas for histories]\label{lem:graph-formulas-for-histories-correctness}
Let $(Q,\mathcal D,B,U)$ be the fixed effectively semialgebraic well-formed core datum. For every admissible prior $h \in Q$ and every history $\pi$, the following hold.
\begin{enumerate}
\item $\mathbb{R} \models \operatorname{Post}_\pi(h,\bar{u})$ if and only if $\operatorname{up}_{\pi}(h)$ is defined and $\bar{u}=\operatorname{up}_{\pi}(h)$.
\item $\mathbb{R} \models \operatorname{Prob}_\pi(h,z)$ if and only if $z=r_\pi(h)$.
\end{enumerate}
\end{lemma}

\begin{proof}
We prove parts (i) and (ii) simultaneously by induction on the length of $\pi$.

Base case $\pi \equiv \varepsilon$. Part (i) follows immediately from $\operatorname{up}_{\varepsilon}(h)=h$ and the definition of $\operatorname{Post}_\varepsilon$. Part (ii) follows immediately from $r_\varepsilon(h)=1$ and the definition of $\operatorname{Prob}_\varepsilon$.

Induction step. Assume that $\pi$ is nonempty, so $\pi \equiv \pi';(D,j)$, and that parts (i) and (ii) hold for $\pi'$.

\medskip

\noindent\textit{Proof of part (i).}

\smallskip

\noindent\textit{Forward implication.}
Suppose that $\mathbb{R} \models \operatorname{Post}_{\pi';(D,j)}(h,\bar{v})$. Then there exists $\bar{u}$ such that $\mathbb{R} \models \operatorname{Post}_{\pi'}(h,\bar{u})$ and $\mathbb{R} \models \operatorname{Step}_{D,j}(\bar{u},\bar{v})$. Unfolding the definition of $\operatorname{Step}_{D,j}$, we have $\Theta_Q(\bar{u})$, $\operatorname{Born}_{D,j}(\bar{u})>0$, $\Theta_Q(\bar{v})$, and $v_m \cdot \operatorname{Born}_{D,j}(\bar{u})=\operatorname{Num}^{D,j}_m(\bar{u})$ for every $m \in I$. By the induction hypothesis for part (i), $\operatorname{up}_{\pi'}(h)$ is defined and $\bar{u}=\operatorname{up}_{\pi'}(h)$. Hence $\operatorname{Born}_{D,j}(\bar{u})=q_j^D(\operatorname{up}_{\pi'}(h))>0$, so $\operatorname{up}_{\pi';(D,j)}(h)$ is defined. Since $\operatorname{Born}_{D,j}(\bar{u})>0$, the coordinate equations in the definition of $\operatorname{Step}_{D,j}$ yield
$v_m=\operatorname{Num}^{D,j}_m(\bar{u})/\operatorname{Born}_{D,j}(\bar{u})$
for every $m\in I$. By the definition of the one-step update and the equality $\bar{u}=\operatorname{up}_{\pi'}(h)$, it follows that
$v_m=\operatorname{up}_{D,j}(\operatorname{up}_{\pi'}(h))_m$
for every $m\in I$. Therefore $\bar{v}=\operatorname{up}_{\pi';(D,j)}(h)$.
\smallskip

\noindent\textit{Reverse implication.}
Suppose that $\operatorname{up}_{\pi';(D,j)}(h)$ is defined. Then $\operatorname{up}_{\pi'}(h)$ is defined and $q_j^D(\operatorname{up}_{\pi'}(h))>0$. By the induction hypothesis for part (i), $\mathbb{R} \models \operatorname{Post}_{\pi'}(h,\operatorname{up}_{\pi'}(h))$. By Lemma~\ref{lem:history-closure}, $\operatorname{up}_{\pi'}(h) \in Q$, so $\mathbb{R} \models \Theta_Q(\operatorname{up}_{\pi'}(h))$. Since the datum is well-formed, $\operatorname{up}_{\pi';(D,j)}(h) \in Q$, so $\mathbb{R} \models \Theta_Q(\operatorname{up}_{\pi';(D,j)}(h))$. By the recursive definition of updates, $\operatorname{up}_{\pi';(D,j)}(h)=\operatorname{up}_{D,j}(\operatorname{up}_{\pi'}(h))$. For every $m \in I$, the definition of the one-step update gives $\operatorname{up}_{\pi';(D,j)}(h)_m=\bigl(\sum_{i \in I}\alpha_i(\operatorname{up}_{\pi'}(h))U^{D,j}(m|i)\bigr)/q_j^D(\operatorname{up}_{\pi'}(h))$. Since $q_j^D(\operatorname{up}_{\pi'}(h))>0$, it follows that $\operatorname{up}_{\pi';(D,j)}(h)_m \cdot q_j^D(\operatorname{up}_{\pi'}(h))=\sum_{i \in I}\alpha_i(\operatorname{up}_{\pi'}(h))U^{D,j}(m|i)$. This is exactly the coordinate condition in $\operatorname{Step}_{D,j}$. Hence $\mathbb{R} \models \operatorname{Post}_{\pi';(D,j)}(h,\operatorname{up}_{\pi';(D,j)}(h))$.

\medskip

\noindent\textit{Proof of part (ii).}

\smallskip

\noindent\textit{Forward implication.}
Suppose that $\mathbb{R} \models \operatorname{Prob}_{\pi';(D,j)}(h,z')$. Then there exists $z$ such that $\mathbb{R} \models \operatorname{Prob}_{\pi'}(h,z)$ and the disjunction in the recursive clause for $\operatorname{Prob}_{\pi';(D,j)}$ holds. By the induction hypothesis for part (ii), $z=r_{\pi'}(h)$. If $z=0$, then $r_{\pi'}(h)=0$. By Lemma~\ref{lem:definability}, the history $\pi'$ is not executable. Hence $r_{\pi';(D,j)}(h)=0$, and the formula forces $z'=0$. Suppose instead that $z>0$. The second disjunct supplies a tuple $\bar{u}$ such that $\mathbb{R} \models \operatorname{Post}_{\pi'}(h,\bar{u})$. By the induction hypothesis for part (i), $\bar{u}=\operatorname{up}_{\pi'}(h)$. The inner disjunction says either $\operatorname{Born}_{D,j}(\bar{u})=0$ and $z'=0$, or $\operatorname{Born}_{D,j}(\bar{u})>0$ and $z'=z\cdot\operatorname{Born}_{D,j}(\bar{u})$. Since $\operatorname{Born}_{D,j}(\operatorname{up}_{\pi'}(h))=q_j^{D,\pi'}(h)$, these alternatives give exactly $z'=r_{\pi';(D,j)}(h)$.

\smallskip

\noindent\textit{Reverse implication.}
Suppose that $z'=r_{\pi';(D,j)}(h)$, and set $z:=r_{\pi'}(h)$. By the induction hypothesis for part (ii), $\mathbb{R} \models \operatorname{Prob}_{\pi'}(h,z)$. If $z=0$, then $z'=0$, and the first disjunct in the recursive clause for $\operatorname{Prob}_{\pi';(D,j)}$ applies. Suppose instead that $z>0$. By Lemma~\ref{lem:definability}, $\operatorname{up}_{\pi'}(h)$ is defined. The induction hypothesis for part (i) gives $\mathbb{R} \models \operatorname{Post}_{\pi'}(h,\operatorname{up}_{\pi'}(h))$. By Lemma~\ref{lem:history-closure}, $\operatorname{up}_{\pi'}(h)\in Q$, so $\operatorname{Born}_{D,j}(\operatorname{up}_{\pi'}(h))\ge0$. If $\operatorname{Born}_{D,j}(\operatorname{up}_{\pi'}(h))=0$, then $z'=0$, and the first alternative in the inner disjunction applies. If $\operatorname{Born}_{D,j}(\operatorname{up}_{\pi'}(h))>0$, then $z'=z\cdot\operatorname{Born}_{D,j}(\operatorname{up}_{\pi'}(h))$, and the second alternative applies. Therefore $\mathbb{R} \models \operatorname{Prob}_{\pi';(D,j)}(h,z')$.

This completes the simultaneous induction.
\end{proof}

After defining the graph formulas for histories, we turn to the values of terms in $\mathcal L_0$. We first introduce formulas for executability and for the totalized values of posterior coordinates and posterior outcome probabilities.

\begin{definition}[Executability and posterior value formulas]\label{def:executability-posterior-value-formulas}
For each history $\pi$, define
\begin{itemize}
\item $\operatorname{Exec}_\pi(\bar{x}) := \exists \bar{u}\,\operatorname{Post}_\pi(\bar{x},\bar{u})$.
\item $\operatorname{Coord}_i^\pi(\bar{x},y) := \bigl(\exists \bar{u}(\operatorname{Post}_\pi(\bar{x},\bar{u}) \wedge y=u_i)\bigr) \vee \bigl(\neg \operatorname{Exec}_\pi(\bar{x}) \wedge y=0\bigr)$, for $i \in I$.
\item $\operatorname{Out}_{D,j}^\pi(\bar{x},y) := \bigl(\exists \bar{u}(\operatorname{Post}_\pi(\bar{x},\bar{u}) \wedge y=\operatorname{Born}_{D,j}(\bar{u}))\bigr) \vee \bigl(\neg \operatorname{Exec}_\pi(\bar{x}) \wedge y=0\bigr)$, for $D \in \mathcal D$ and $j \in O_D$.
\end{itemize}
\end{definition}

The formula $\operatorname{Exec}_\pi(\bar{x})$ expresses that the history $\pi$ is executable at the prior $\bar{x}$. If $\pi$ is executable at $\bar{x}$, $\operatorname{Coord}_i^\pi(\bar{x},y)$ defines $y$ as the $i$th coordinate of the posterior reference prior, and $\operatorname{Out}_{D,j}^\pi(\bar{x},y)$ defines $y$ as the Born probability of outcome $j$ for $D$ evaluated at that posterior reference prior. If $\pi$ is not executable at $\bar{x}$, both formulas define $y$ as $0$, in accordance with the totalization convention in Definition~\ref{def:posterior-quantities}.

The following definition associates with each term $t$ of $\mathcal L_0$ a first-order graph formula $\operatorname{Val}_t(\bar{x},y)$ for the function that maps an admissible prior $\bar{x}$ to the value denoted by $t$.

\begin{definition}[First-order graph formulas for terms]\label{def:graph-formulas-for-terms}
For each term $t$ of $\mathcal L_0$, the formula $\operatorname{Val}_t(\bar{x},y)$ over $\Sigma_K$ is defined by recursion on the formation of $t$ as follows.
\begin{itemize}
\item $\operatorname{Val}_a(\bar{x},y) := (y=a)$, for each constant symbol $a$.
\item $\operatorname{Val}_{h_i}(\bar{x},y) := (y=x_i)$.
\item $\operatorname{Val}_{q_j^D}(\bar{x},y) := (y=\operatorname{Born}_{D,j}(\bar{x}))$.
\item $\operatorname{Val}_{r_\pi}(\bar{x},y) := \operatorname{Prob}_\pi(\bar{x},y)$.
\item $\operatorname{Val}_{h_i^\pi}(\bar{x},y) := \operatorname{Coord}_i^\pi(\bar{x},y)$.
\item $\operatorname{Val}_{q_j^{D,\pi}}(\bar{x},y) := \operatorname{Out}_{D,j}^\pi(\bar{x},y)$.
\item $\operatorname{Val}_{t+u}(\bar{x},y) := \exists y_1 \exists y_2 \bigl(\operatorname{Val}_t(\bar{x},y_1) \wedge \operatorname{Val}_u(\bar{x},y_2) \wedge y=y_1+y_2\bigr)$.
\item $\operatorname{Val}_{-t}(\bar{x},y) := \exists y_1 \bigl(\operatorname{Val}_t(\bar{x},y_1) \wedge y=-y_1\bigr)$.
\item $\operatorname{Val}_{t \cdot u}(\bar{x},y) := \exists y_1 \exists y_2 \bigl(\operatorname{Val}_t(\bar{x},y_1) \wedge \operatorname{Val}_u(\bar{x},y_2) \wedge y=y_1 \cdot y_2\bigr)$.
\end{itemize}
\end{definition}

The formula $\operatorname{Val}_t(\bar{x},y)$ states that $y$ is the value denoted by $t$ at a logical state whose prior component is $\bar{x}$. The record marker does not appear because the denotation of a term depends only on the prior component of a logical state.

\begin{proposition}[Effective constructibility of graph formulas for terms]\label{prop:graph-formulas-for-terms-effective}
There is an effective procedure that, given a term $t$ of $\mathcal L_0$, computes the formula $\operatorname{Val}_t$.
\end{proposition}

\begin{proof}
The formula $\operatorname{Val}_t$ is computed by primitive recursion on the formation of $t$ from Definition~\ref{def:graph-formulas-for-terms}. Proposition~\ref{prop:graph-formulas-for-histories-effective} provides the formulas $\operatorname{Post}_\pi$ and $\operatorname{Prob}_\pi$ effectively. Definition~\ref{def:executability-posterior-value-formulas} then gives $\operatorname{Exec}_\pi$, $\operatorname{Coord}_i^\pi$, and $\operatorname{Out}_{D,j}^\pi$ effectively from $\operatorname{Post}_\pi$. These constructions supply the clauses for the atomic terms indexed by histories.
\end{proof}

\begin{lemma}[Correctness of graph formulas for terms]\label{lem:graph-formulas-for-terms-correctness}
Let $l$ be a record marker, let $h \in Q$, and let $t$ be a term of $\mathcal L_0$. Then $\mathbb{R} \models \operatorname{Val}_t(h,y)$ if and only if $y=\llbracket t \rrbracket_{(h,l)}$.
\end{lemma}

\begin{proof}
We argue by induction on the formation of $t$.

Cases $t \equiv a$, $t \equiv h_i$, and $t \equiv q_j^D$. These cases are immediate from the definitions.

Case $t \equiv r_\pi$. The claim is part (ii) of Lemma~\ref{lem:graph-formulas-for-histories-correctness}.

Case $t \equiv h_i^\pi$.

Subcase $\operatorname{Exec}_\pi(h)$ holds. Then part (i) of Lemma~\ref{lem:graph-formulas-for-histories-correctness} shows that there is a unique $\bar{u}$ with $\operatorname{Post}_\pi(h,\bar{u})$, namely $\bar{u}=\operatorname{up}_{\pi}(h)$. Thus $\operatorname{Coord}_i^\pi(h,y)$ holds if and only if $y=(\operatorname{up}_{\pi}(h))_i=h_i^\pi(h)$.

Subcase $\operatorname{Exec}_\pi(h)$ fails. Then $\operatorname{Coord}_i^\pi(h,y)$ holds if and only if $y=0$, which agrees with the semantic convention for $h_i^\pi(h)$.

Case $t \equiv q_j^{D,\pi}$. The same argument gives two subcases.

Subcase $\operatorname{Exec}_\pi(h)$ holds. Then $\operatorname{Out}_{D,j}^\pi(h,y)$ holds if and only if $y=\operatorname{Born}_{D,j}(\operatorname{up}_{\pi}(h))=q_j^{D,\pi}(h)$.

Subcase $\operatorname{Exec}_\pi(h)$ fails. Then $\operatorname{Out}_{D,j}^\pi(h,y)$ holds if and only if $y=0$, which agrees with the semantic convention for $q_j^{D,\pi}(h)$.

Case $t \equiv u+v$. By definition, $\mathbb{R} \models \operatorname{Val}_{u+v}(h,y)$ if and only if there exist $y_1,y_2$ such that $\operatorname{Val}_{u}(h,y_1)$, $\operatorname{Val}_v(h,y_2)$, and $y=y_1+y_2$. By the induction hypothesis, this is equivalent to $y_1=\llbracket u \rrbracket_{(h,l)}$ and $y_2=\llbracket v \rrbracket_{(h,l)}$, and hence to $y=\llbracket u \rrbracket_{(h,l)}+\llbracket v \rrbracket_{(h,l)}=\llbracket u+v \rrbracket_{(h,l)}$.

Cases $t \equiv -u$ and $t \equiv u \cdot v$. These cases follow in the same way from their defining clauses and the induction hypothesis.

This completes the induction.
\end{proof}

Recall that $\mathrm{Rec}$ denotes the finite set of record markers $\{\bot\} \cup \{(D,j) : D \in \mathcal{D},\ j \in O_D\}$.

The final translation step passes from terms to formulas. Since record markers are finite and external to the ordered field, they are handled by fixing one marker at a time and translating relative to it.

\begin{definition}[First-order translation of the fragment without dynamic operators]\label{def:static-fo-translation}
For each record marker $l \in \mathrm{Rec}$ and each formula $\chi \in \mathcal L_0$, the first-order formula $\operatorname{FO}_l(\chi;\bar{x})$ is defined by recursion on the complexity of $\chi$ as follows.
\begin{itemize}
\item $\operatorname{FO}_l(\mathsf{Last}(D,j);\bar{x}) := (0=0)$, for $l=(D,j)$.
\item $\operatorname{FO}_l(\mathsf{Last}(D,j);\bar{x}) := (0 \neq 0)$ otherwise.
\item $\operatorname{FO}_l(t \ge 0;\bar{x}) := \exists y \bigl(\operatorname{Val}_t(\bar{x},y) \wedge y \ge 0\bigr)$.
\item $\operatorname{FO}_l(\neg \chi;\bar{x}) := \neg \operatorname{FO}_l(\chi;\bar{x})$.
\item $\operatorname{FO}_l(\chi \wedge \psi;\bar{x}) := \operatorname{FO}_l(\chi;\bar{x}) \wedge \operatorname{FO}_l(\psi;\bar{x})$.
\end{itemize}
\end{definition}

For every $h\in Q$, the formula $\operatorname{FO}_l(\chi;\bar{x})$ represents the truth of $\chi$ at the logical state $(h,l)$ under the assignment $\bar{x}=h$. The two clauses for $\mathsf{Last}(D,j)$ are truth constants determined by the fixed record marker $l$, while arithmetic atomic formulas are translated using the formulas $\operatorname{Val}_t$ for their terms.

\begin{proposition}[Effective constructibility of the first-order translation]\label{prop:static-translation-effective}
There is an effective procedure that, given a record marker $l$ and a formula $\chi \in \mathcal L_0$, computes the formula $\operatorname{FO}_l(\chi;\bar{x})$ and the sentence $\forall \bar{x}(\Theta_Q(\bar{x}) \to \operatorname{FO}_l(\chi;\bar{x}))$.
\end{proposition}

\begin{proof}
The formula $\operatorname{FO}_l(\chi;\bar{x})$ is computed by primitive recursion on the complexity of $\chi$, using Proposition~\ref{prop:graph-formulas-for-terms-effective} in the arithmetic atomic case. Its universal closure with antecedent $\Theta_Q$ is then obtained effectively.
\end{proof}

The next lemma relates truth in $\mathcal L_0$ at a logical state to the corresponding first-order formula.

\begin{lemma}[Correctness of the first-order translation]\label{lem:static-translation-correctness}
Let $\chi \in \mathcal L_0$, let $l \in \mathrm{Rec}$, and let $h \in Q$. Then $(h,l) \models \chi$ if and only if $\mathbb{R} \models \operatorname{FO}_l(\chi;h)$.
\end{lemma}

\begin{proof}
We argue by induction on the complexity of $\chi$.

Case $\chi \equiv \mathsf{Last}(D,j)$. The claim is immediate from Definition~\ref{def:static-fo-translation}.

Case $\chi \equiv (t \ge 0)$. By the semantic clause for arithmetic atoms, $(h,l) \models (t \ge 0)$ if and only if $\llbracket t \rrbracket_{(h,l)} \ge 0$. Lemma~\ref{lem:graph-formulas-for-terms-correctness} makes this condition equivalent to the existence of a real number $y$ such that $\operatorname{Val}_t(h,y)$ and $y \ge 0$. By Definition~\ref{def:static-fo-translation}, this is equivalent to $\mathbb{R} \models \operatorname{FO}_l(\chi;h)$.

Cases $\chi \equiv \neg\psi$ and $\chi \equiv \psi \wedge \theta$. These cases follow immediately from the induction hypothesis and the classical clauses in both semantics.
\end{proof}

Consequently, validity in $\mathcal L_0$ is equivalent to a finite family of first-order sentences over real closed fields.

\begin{corollary}[First-order characterization of validity in $\mathcal L_0$]\label{cor:first-order-characterization-static-validity}
For every formula $\chi \in \mathcal L_0$, the following are equivalent.
\begin{enumerate}
\item $\models_{(Q,\mathcal D,B,U)} \chi$.
\item For every record marker $l \in \mathrm{Rec}$, one has $\mathbb{R} \models \forall \bar{x} \bigl(\Theta_Q(\bar{x}) \to \operatorname{FO}_l(\chi;\bar{x})\bigr)$.
\end{enumerate}
\end{corollary}

\begin{proof}
By Definition~\ref{def:validity-over-well-formed-core-datum}, validity over the datum means truth at every state $(h,l)$ with $h \in Q$ and $l \in \mathrm{Rec}$. By Lemma~\ref{lem:static-translation-correctness}, this is equivalent to the truth in $\mathbb{R}$ of $\operatorname{FO}_l(\chi;h)$ for every such pair $(h,l)$. Since $Q$ is defined by $\Theta_Q$, the result follows.
\end{proof}

\subsection{Effective calculus}\label{subsec:effective-calculus}

The effective calculus rests on the transfer of first-order truth between $K$ and $\mathbb{R}$ and on the decidability of the resulting theory. By model completeness of the theory of real closed fields, the inclusion $K\subseteq\mathbb{R}$ is elementary. Effective quantifier elimination, together with the effective presentation of $K$, reduces each sentence over $\Sigma_K$ to a quantifier-free sentence whose truth is decidable from the field operations and order on $K$. Hence the complete first-order theory $\mathrm{Th}_{\Sigma_K}(\mathbb{R})$ is decidable \cite{BasuPollackRoy2006,Marker2002}. Fix once and for all a recursive first-order proof system $\mathsf{RCF}_{K}$ that is sound and complete for that theory.

\begin{definition}[Validity sentences]\label{def:validity-sentences}
For $\chi \in \mathcal L_0$ and a record marker $l \in \mathrm{Rec}$, set $\operatorname{Sent}_l(\chi) := \forall \bar{x}\bigl(\Theta_Q(\bar{x}) \to \operatorname{FO}_l(\chi;\bar{x})\bigr)$.
\end{definition}

The sentence $\operatorname{Sent}_l(\chi)$ expresses validity of $\chi$ at the fixed record marker $l$ over the admissible prior space.

The preceding constructions determine a calculus for the object language. Its definition combines propositional reasoning, reduction from $\mathcal L$ to $\mathcal L_0$, and external first-order proofs over real closed fields.

\begin{definition}[Recursive QBism calculus]\label{def:recursive-qbism-calculus}
A \emph{QBism proof over the fixed effectively semialgebraic datum} is a finite annotated sequence of formulas from $\mathcal L$ in which each line is one of the following.
\begin{enumerate}
\item A substitution instance of a propositional tautology.
\item A reduction biconditional $\varphi \leftrightarrow \operatorname{Red}(\varphi)$.
\item A formula $\chi \in \mathcal L_0$ together with, for every record marker $l \in \mathrm{Rec}$, a finite $\mathsf{RCF}_{K}$ proof of the sentence $\operatorname{Sent}_l(\chi)$.
\item A formula obtained from earlier formulas by modus ponens.
\end{enumerate}
We write $\mathsf{QB}_{\mathfrak M} \vdash \varphi$ if there exists such a proof whose final formula is $\varphi$.
\end{definition}

\medskip

\begin{remark}[Recursive proof checking]
Because $\mathrm{Rec}$ is finite, Propositions~\ref{prop:reduction-effective} and \ref{prop:static-translation-effective} provide effective procedures for constructing every reduction biconditional and every sentence $\operatorname{Sent}_l(\chi)$ required at a proof line. The remaining verification consists of checking propositional instances, applications of modus ponens, and the attached $\mathsf{RCF}_{K}$ proofs, each of which is recursive. Therefore every finite annotated sequence can be checked recursively, so proof checking for $\mathsf{QB}_{\mathfrak M}$ is recursive.
\end{remark}

\subsection{Metatheoretic consequences}\label{subsec:metatheoretic-consequences}

The global reduction theorem moves formulas from $\mathcal L$ into $\mathcal L_0$, and the first-order translation sends $\mathcal L_0$ into formulas over real closed fields. The calculus $\mathsf{QB}_{\mathfrak M}$ incorporates both operations, yielding the soundness and completeness theorem below.

\begin{theorem}[Recursive soundness and completeness]\label{thm:recursive-completeness}
Let $(Q,\mathcal D,B,U)$ be an effectively semialgebraic well-formed core datum over $K$ with fixed defining formula $\Theta_Q$, and let $\mathfrak M:=(Q,\mathcal D,B,U,\Theta_Q,K)$. Then, for every formula $\varphi \in \mathcal L$, one has $\mathsf{QB}_{\mathfrak M} \vdash \varphi$ if and only if $\models_{(Q,\mathcal D,B,U)} \varphi$.
\end{theorem}

\begin{proof}
We first prove soundness. A line admitted under clause (i) of Definition~\ref{def:recursive-qbism-calculus} is valid because the Boolean clauses of the semantics are classical. A line admitted under clause (ii) is valid by Theorem~\ref{thm:global-reduction}. For clause (iii), let $\chi \in \mathcal L_0$ be accompanied by $\mathsf{RCF}_{K}$ proofs of all sentences $\operatorname{Sent}_l(\chi)$ with $l \in \mathrm{Rec}$. The soundness of $\mathsf{RCF}_{K}$ implies that every such sentence is true in $\mathbb{R}$. Corollary~\ref{cor:first-order-characterization-static-validity} then yields $\models_{(Q,\mathcal D,B,U)} \chi$. Finally, clause (iv) preserves validity because modus ponens does. Therefore every theorem of $\mathsf{QB}_{\mathfrak M}$ is valid over the datum.

We next prove completeness. Assume $\models_{(Q,\mathcal D,B,U)} \varphi$. By Theorem~\ref{thm:global-reduction}, the formula $\operatorname{Red}(\varphi)$ is valid, and Lemma~\ref{lem:static-fragment-reduction} ensures that $\operatorname{Red}(\varphi) \in \mathcal L_0$. By Corollary~\ref{cor:first-order-characterization-static-validity}, the sentence $\operatorname{Sent}_l(\operatorname{Red}(\varphi))$ is true in $\mathbb{R}$ for every record marker $l \in \mathrm{Rec}$. Since $\mathsf{RCF}_{K}$ is complete for $\mathrm{Th}_{\Sigma_K}(\mathbb{R})$, each of these sentences has a proof in $\mathsf{RCF}_{K}$. Clause (iii) of Definition~\ref{def:recursive-qbism-calculus} therefore permits the introduction of $\operatorname{Red}(\varphi)$.

Clause (ii) permits the introduction of the reduction biconditional $\varphi \leftrightarrow \operatorname{Red}(\varphi)$. Clause (i) permits the substitution instance $(\varphi \leftrightarrow \operatorname{Red}(\varphi)) \to (\operatorname{Red}(\varphi) \to \varphi)$ of the propositional tautology $(\alpha \leftrightarrow \beta)\to(\beta\to\alpha)$. Two applications of modus ponens under clause (iv) then yield $\varphi$. Hence $\mathsf{QB}_{\mathfrak M} \vdash \varphi$.
\end{proof}

The theorem yields decidability immediately, because the ambient first-order theory is decidable and all reductions involved are effective. The following corollary states this consequence explicitly.

\begin{corollary}[Decidability of validity]\label{cor:decidability}
Let $(Q,\mathcal D,B,U)$ be an effectively semialgebraic well-formed core datum over $K$ with fixed defining formula $\Theta_Q$, and let $\mathfrak M:=(Q,\mathcal D,B,U,\Theta_Q,K)$. Then the set $\{\varphi \in \mathcal L \mid \mathsf{QB}_{\mathfrak M} \vdash \varphi\}$ is decidable. Equivalently, validity over the datum $(Q,\mathcal D,B,U)$ is decidable.
\end{corollary}

\begin{proof}
Given a formula $\varphi \in \mathcal L$, compute $\operatorname{Red}(\varphi)$ by Proposition~\ref{prop:reduction-effective}. By Lemma~\ref{lem:static-fragment-reduction}, this belongs to $\mathcal L_0$. By Proposition~\ref{prop:static-translation-effective}, for every record marker $l \in \mathrm{Rec}$ one can effectively compute the sentence $\operatorname{Sent}_l(\operatorname{Red}(\varphi))$.

By Corollary~\ref{cor:first-order-characterization-static-validity} and Theorem~\ref{thm:global-reduction}, the formula $\varphi$ is valid over the datum if and only if, for every record marker $l \in \mathrm{Rec}$, the sentence $\operatorname{Sent}_l(\operatorname{Red}(\varphi))$ is true in $\mathbb{R}$. Because $\mathrm{Th}_{\Sigma_K}(\mathbb{R})$ is decidable, one can decide each of these finitely many sentences effectively. Therefore validity of $\varphi$ is decidable. The equivalence with theoremhood in $\mathsf{QB}_{\mathfrak M}$ follows from Theorem~\ref{thm:recursive-completeness}.
\end{proof}

\medskip

\begin{remark}[Scope of the metatheorem]\label{rem:scope-recursive-theorem}
Theorem~\ref{thm:recursive-completeness} gives the main metatheoretic consequence of the preceding reductions. The central structural point is that, for effectively semialgebraic data, validity in the guarded language is reducible to first-order reasoning over real closed fields. Recursive completeness then follows from the completeness of the ambient real closed field proof system.
\end{remark}

\section{Quantum realization}\label{sec:quantum-realization}

\subsection{Quantum data in SIC coordinates}\label{subsec:quantum-realization-sic-reference-data}

In this section we show that, whenever a \emph{symmetric informationally complete (SIC) reference measurement} exists in the chosen dimension, standard quantum theory in finite dimensions provides a realization of the semantic framework developed in Sections~\ref{sec:semantic-data} and~\ref{sec:guarded-dynamic-language}. Under the additional coefficient assumptions stated in Subsection~\ref{subsec:effective-semialgebraic-quantum-data}, the resulting datum also belongs to the effective class studied in Section~\ref{sec:metatheory}.

Fix a dimension $d \ge 2$ in which a SIC exists, together with a Hilbert space $\mathcal H \cong \mathbb C^d$ and projections $\Pi_1,\dots,\Pi_{d^2}$. Assume that each $\Pi_i$ has rank one and that the operators $H_i := \Pi_i/d$ form a \emph{symmetric informationally complete POVM} \cite{RenesBlumeKohoutScottCaves2004,Zauner2011}. Each $H_i$ is positive, and $\sum_{i \in I} H_i = I_{\mathcal H}$. Informational completeness means that the family $\{H_i\}_{i \in I}$ spans the real vector space of Hermitian operators. The projections satisfy $\operatorname{tr}(\Pi_i \Pi_k)=1$ when $i=k$ and $\operatorname{tr}(\Pi_i \Pi_k)=1/(d+1)$ when $i \neq k$. We use the same index set $I = \{1,\dots,d^2\}$ as in the abstract framework.

We also consider a finite family $\mathcal D$ of actual measurements. For each $D \in \mathcal D$, fix a POVM $D = \{D_j\}_{j \in O_D}\!$ on $\mathcal H$ and a \emph{quantum instrument} $\mathcal I^D = \{\mathcal I^D_j\}_{j \in O_D}$ with the same outcome set $O_D$, in the standard sense of quantum measurement theory \cite{DaviesLewis1970,Holevo2012}.

We assume that each map $\mathcal I^D_j$ is completely positive and trace nonincreasing, that $\sum_{j \in O_D} \mathcal I^D_j$ is trace preserving, and that the instrument is associated with the POVM $D$ in the sense that $\operatorname{tr}(\mathcal I^D_j(X))=\operatorname{tr}(X D_j)$ holds for every positive operator $X$. In the present framework, the actual measurement type $D$ carries both a current outcome probability assignment and an update mechanism, and the instrument specifies the latter. The assumptions are stated in this explicit form to make the connection with the abstract kernel conditions transparent.

\medskip

\begin{remark}[Unitary evolution as a singleton outcome]\label{rem:unitary-evolution-singleton}
Unitary time evolution can be represented inside the same formalism. If $V$ is a unitary operator, take an actual measurement with singleton outcome set $O_D=\{0\}$, effect $D_0=I_{\mathcal H}$, and instrument map $\mathcal I^D_0(X)=VXV^*$. Then $B^D(0|i)=1$ for every $i$, the corresponding history step is always executable, and the update sends a state to its unitary image in SIC coordinates. Thus singleton measurements encode this evolution within the same framework.
\end{remark}

The quantum data fixed above determine a candidate for a well-formed core datum. The Born kernels and update kernels will be read off from the POVMs and instruments. The admissible prior space will be obtained by expressing density operators in SIC coordinates. The resulting datum will be denoted by $(Q_{\mathrm{SIC}},\mathcal D,B,U)$.

\subsection{Realization of the kernels}\label{subsec:quantum-realization-kernels}

We begin with the kernels determined by the quantum data.

\begin{definition}[Quantum Born kernels]\label{def:quantum-born-kernels}
For each $D \in \mathcal D$, each $i \in I$, and each $j \in O_D$, set $B^D(j|i) := \operatorname{tr}(D_j \Pi_i)$.
\end{definition}

\begin{definition}[Quantum update kernels]\label{def:quantum-update-kernels}
For each $D \in \mathcal D$, each $j \in O_D$, and each $i,m \in I$, set $U^{D,j}(m|i) := \operatorname{tr}(H_m\,\mathcal I^D_j(\Pi_i))$.
\end{definition}

The value $B^D(j|i)$ is the quantum probability of obtaining the actual outcome $j$ when the input state is the projection $\Pi_i$, which has rank one. The value $U^{D,j}(m|i)$ is the joint weight of obtaining the instrument outcome $j$ and then the SIC reference outcome $m$ from the unnormalized state $\mathcal I^D_j(\Pi_i)$. Summing this joint weight over $m$ gives $B^D(j|i)$, as verified below.

The following proposition verifies that these formulas satisfy the abstract kernel requirements.

\begin{proposition}[Kernel conditions from quantum data]\label{prop:quantum-realization-kernels}
For every $D \in \mathcal D$, the family $B^D(j|i)$ satisfies the Born kernel conditions of Subsection~\ref{subsec:semantic-data-kernels}. For every $D \in \mathcal D$ and every $j \in O_D$, the family $U^{D,j}(m|i)$ also satisfies the required update kernel conditions.
\end{proposition}

\begin{proof}
Let $D \in \mathcal D$, $i \in I$, and $j \in O_D$. Since $D_j$ and $\Pi_i$ are positive operators, the operator $D_j^{1/2}\Pi_i D_j^{1/2}$ is positive, so its trace is nonnegative. Therefore $B^D(j|i) = \operatorname{tr}(D_j \Pi_i) \ge 0$.

To prove normalization, we compute $\sum_{j \in O_D} B^D(j|i) = \sum_{j \in O_D} \operatorname{tr}(D_j \Pi_i) = \operatorname{tr}\bigl((\sum_{j \in O_D} D_j)\Pi_i\bigr) = \operatorname{tr}(I_{\mathcal H}\Pi_i) = \operatorname{tr}(\Pi_i) = 1$, because $D$ is a POVM and each $\Pi_i$ has rank one. Hence $j \mapsto B^D(j|i)$ is a probability distribution on $O_D$.

Now let $m \in I$ as well. Since $\mathcal I^D_j$ is completely positive, it sends the positive operator $\Pi_i$ to the positive operator $\mathcal I^D_j(\Pi_i)$. Since $H_m$ is positive, the operator $H_m^{1/2}\mathcal I^D_j(\Pi_i)H_m^{1/2}$ is positive, so its trace is nonnegative. Therefore $U^{D,j}(m|i) = \operatorname{tr}(H_m\,\mathcal I^D_j(\Pi_i)) \ge 0$.

For the normalization identity, $\sum_{m \in I} U^{D,j}(m|i)=\operatorname{tr}\bigl((\sum_{m \in I}H_m)\mathcal I^D_j(\Pi_i)\bigr)$. Since $\sum_{m \in I}H_m=I_{\mathcal H}$, this value is $\operatorname{tr}(\mathcal I^D_j(\Pi_i))$. The instrument outcome map $\mathcal I^D_j$ has effect $D_j$, so $\operatorname{tr}(\mathcal I^D_j(X))=\operatorname{tr}(XD_j)$ for every positive operator $X$. Applying this identity to $X=\Pi_i$ gives $\operatorname{tr}(\mathcal I^D_j(\Pi_i))=\operatorname{tr}(D_j\Pi_i)=B^D(j|i)$. Therefore $\sum_{m \in I}U^{D,j}(m|i)=B^D(j|i)$, as required.
\end{proof}

\subsection{Realization of the admissible prior space}\label{subsec:quantum-admissible-prior-space}

Let $\rho$ be a density operator on $\mathcal H$. We associate to $\rho$ its \emph{SIC coordinate vector} $h(\rho) := (h_i(\rho))_{i \in I}$, where $h_i(\rho) := \operatorname{tr}(\rho H_i)$. Since each $H_i$ is positive and $\sum_{i \in I} H_i = I_{\mathcal H}$, one has $h_i(\rho) \ge 0$ for every $i \in I$ and $\sum_{i \in I} h_i(\rho) = \operatorname{tr}(\rho) = 1$. Hence $h(\rho) \in \Delta(I)$.

\begin{definition}[Quantum admissible prior space]\label{def:quantum-admissible-prior-space}
The set $Q_{\mathrm{SIC}} := \{h(\rho) \mid \rho \text{ is a density operator on } \mathcal H\}$ is the \emph{quantum admissible prior space}.
\end{definition}

Geometrically, $Q_{\mathrm{SIC}}$ is the \emph{SIC image} of the quantum state space under the SIC coordinate map. Because the SIC is informationally complete, the map $\rho \mapsto h(\rho)$ is injective \cite{RenesBlumeKohoutScottCaves2004}. Thus $Q_{\mathrm{SIC}}$ is a faithful probabilistic representation of the usual quantum state space.

We use the standard \emph{SIC reconstruction formula}. If $h = h(\rho)$, then $\rho = \sum_{i \in I} \Bigl((d+1)h_i - (1/d)\Bigr)\Pi_i$. In the present notation, this reads $\rho = \sum_{i \in I} \alpha_i(h)\Pi_i$. This identity underlies the affine form of the Born rule in SIC coordinates \cite{FuchsSchack2013}.

\subsection{Positive semidefiniteness criterion for the SIC image}\label{subsec:quantum-realization-direct-description}

The SIC reconstruction formula allows one to describe the quantum admissible prior space directly inside the simplex. This description clarifies the geometry of the quantum admissible prior space and will also be used in the proof of effective semialgebraicity.

\begin{lemma}[Characterization of the SIC image]\label{lem:sic-image-characterization}
Assume that a SIC exists in dimension $d$, and let $\Pi_1,\dots,\Pi_{d^2}$ and $H_i:=\Pi_i/d$ be fixed as above. For $h=(h_i)_{i \in I} \in \mathbb{R}^{d^2}$, define $\rho(h):=\sum_{i \in I}\bigl((d+1)h_i-(1/d)\bigr)\Pi_i$. Then the following are equivalent.
\begin{enumerate}
\item $h \in Q_{\mathrm{SIC}}$.
\item $h \in \Delta(I)$ and $\rho(h)$ is positive semidefinite.
\end{enumerate}
\end{lemma}

\begin{proof}
Assume first that $h \in Q_{\mathrm{SIC}}$. Then $h=h(\rho)$ for some density operator $\rho$. By the SIC reconstruction formula, $\rho=\rho(h)$. Hence $\rho(h)$ is positive semidefinite. Since $h=h(\rho)$ is a probability vector, one also has $h \in \Delta(I)$.

Conversely, assume that $h \in \Delta(I)$ and that $\rho(h)$ is positive semidefinite. Since each $\Pi_i$ is Hermitian and the coefficients $\bigl((d+1)h_i-(1/d)\bigr)$ are real, the operator $\rho(h)$ is Hermitian. Moreover, $\operatorname{tr}(\rho(h))=\sum_{i \in I}\bigl((d+1)h_i-(1/d)\bigr)\operatorname{tr}(\Pi_i)=(d+1)\sum_{i \in I} h_i-|I|/d=(d+1)-d=1$, because $|I|=d^2$ and $\operatorname{tr}(\Pi_i)=1$ for every $i$. Thus $\rho(h)$ is a density operator.

It remains to show that $h=h(\rho(h))$. Let $k \in I$. Then $\operatorname{tr}(\rho(h)H_k)=(1/d)\sum_{i \in I}\bigl((d+1)h_i-(1/d)\bigr)\operatorname{tr}(\Pi_i\Pi_k)$. Writing $\alpha_i=\alpha_i(h)$, using the SIC overlap relations, and separating the $i=k$ term, we obtain $\operatorname{tr}(\rho(h)H_k)=(1/d)\Bigl(\alpha_k+(1/(d+1))\sum_{i \neq k} \alpha_i\Bigr)=(1/d)\Bigl(\alpha_k+(1/(d+1))\bigl(\sum_{i \in I} \alpha_i-\alpha_k\bigr)\Bigr)$. Since $h \in \Delta(I)$, one has $\sum_{i \in I} \alpha_i=(d+1)\sum_{i \in I} h_i-|I|/d=(d+1)-d=1$. Therefore $\operatorname{tr}(\rho(h)H_k)=(1/d)\Bigl(\alpha_k+(1-\alpha_k)/(d+1)\Bigr)=(1/d)\Bigl((d/(d+1))\alpha_k+(1/(d+1))\Bigr)$. Substituting $\alpha_k=(d+1)h_k-1/d$ yields $\operatorname{tr}(\rho(h)H_k)=h_k$. Since this holds for every $k \in I$, we conclude that $h=h(\rho(h))$. Hence $h \in Q_{\mathrm{SIC}}$.
\end{proof}

\subsection{Qplex geometry of the SIC image}\label{subsec:quantum-realization-qplex-geometry}

The SIC image of the quantum state space satisfies the standard qplex geometry conditions. These conditions supplement the dynamic semantics with geometric constraints on the admissible prior space.

\begin{definition}[Qplex geometry conditions]\label{def:qplex-geometry-conditions}
Let $(Q,\mathcal{D},B,U)$ be a well-formed core datum. We say that $Q$ \emph{satisfies the qplex geometry conditions} if it satisfies the consistency bounds $L_d \le \langle h,s\rangle \le U_d$ for all $h,s \in Q$ and the lower polar condition $Q=Q^\sharp$.
\end{definition}

\begin{proposition}[Qplex properties of the SIC image]\label{prop:quantum-realization-qplex-geometry}
Assume that a SIC exists in dimension $d$. Then the SIC image $Q_{\mathrm{SIC}}$ of the quantum state space is nonempty, satisfies the consistency bounds $L_d \le \langle h,s\rangle \le U_d$ for all $h,s \in Q_{\mathrm{SIC}}$, and satisfies the lower polar condition $Q_{\mathrm{SIC}}=Q_{\mathrm{SIC}}^\sharp$ under the normalization adopted in the present paper.
\end{proposition}

\begin{proof}
Nonemptiness is immediate, since every density operator yields an element of $Q_{\mathrm{SIC}}$.

The proof uses the following basic trace identity. Let $x \in \Delta(I)$, let $\tau$ be a density operator, and put $s=h(\tau)$. Since $s_i=\operatorname{tr}(\tau H_i)=\operatorname{tr}(\tau\Pi_i)/d$, one has
$\operatorname{tr}(\rho(x)\tau)=\sum_{i \in I}\alpha_i(x)\operatorname{tr}(\Pi_i\tau)=d\sum_{i \in I}\alpha_i(x)s_i=d(d+1)\langle x,s\rangle-1$. The last equality uses $\sum_{i \in I}s_i=1$, since $s\in\Delta(I)$.

Let $h,s \in Q_{\mathrm{SIC}}$. Choose density operators $\rho$ and $\tau$ with $h=h(\rho)$ and $s=h(\tau)$. By the SIC reconstruction formula, $\rho(h)=\rho$. Applying the identity with $x=h$ gives $\operatorname{tr}(\rho\tau)=d(d+1)\langle h,s\rangle-1$. Since $\rho$ and $\tau$ are positive semidefinite, $\operatorname{tr}(\rho\tau)\ge0$. Also, by the Hilbert--Schmidt Cauchy--Schwarz inequality and the inequalities $\operatorname{tr}(\rho^2)\le1$ and $\operatorname{tr}(\tau^2)\le1$ for density operators, one has $\operatorname{tr}(\rho\tau)\le1$. Therefore $0\le d(d+1)\langle h,s\rangle-1\le1$, which is equivalent to $L_d\le\langle h,s\rangle\le U_d$.

It remains to prove the lower polar condition. The inclusion $Q_{\mathrm{SIC}}\subseteq Q_{\mathrm{SIC}}^\sharp$ follows from the lower consistency bound just proved. Conversely, let $x\in Q_{\mathrm{SIC}}^\sharp$. Then $x\in\Delta(I)$ and $L_d\le\langle x,s\rangle$ for every $s\in Q_{\mathrm{SIC}}$. For every density operator $\tau$, putting $s=h(\tau)$ in the trace identity gives $\operatorname{tr}(\rho(x)\tau)\ge0$. In particular, for every unit vector $v$, taking $\tau=|v\rangle\langle v|$ yields $\langle v,\rho(x)v\rangle\ge0$. Hence $\rho(x)$ is positive semidefinite. By Lemma~\ref{lem:sic-image-characterization}, this implies $x\in Q_{\mathrm{SIC}}$. Thus $Q_{\mathrm{SIC}}^\sharp\subseteq Q_{\mathrm{SIC}}$, and consequently $Q_{\mathrm{SIC}}=Q_{\mathrm{SIC}}^\sharp$.
\end{proof}

\medskip

\begin{remark}[Normalization of the qplex geometry conditions]\label{rem:quantum-realization-qplex-normalization}
The proof above fixes the normalization directly from the SIC reconstruction formula. The qplex literature supplies the terminology and the geometric viewpoint, while the argument derives the consistency bounds and lower polar condition for the SIC image in the notation of this paper.
\end{remark}

\subsection{Agreement with quantum instruments}\label{subsec:quantum-instrument-agreement}

We now compare the abstract Born and update formulas with their quantum counterparts. The following propositions show that the abstract kernels reproduce ordinary quantum probabilities and instrument updates in SIC coordinates.

\begin{proposition}[Agreement of the Born probabilities]\label{prop:quantum-realization-born}
Let $\rho$ be a density operator and let $h = h(\rho)$. Then for every $D \in \mathcal D$ and every $j \in O_D$, one has $q_j^D(h) = \operatorname{tr}(\rho D_j)$.
\end{proposition}

\begin{proof}
By definition of the core Born expression and of the kernel $B^D(j|i)$, we have $q_j^D(h) = \sum_{i \in I} \alpha_i(h) B^D(j|i) = \sum_{i \in I} \alpha_i(h)\operatorname{tr}(D_j \Pi_i)$. By linearity of the trace, this equals $\operatorname{tr}\bigl(D_j \sum_{i \in I} \alpha_i(h)\Pi_i\bigr)$. By the SIC reconstruction formula, $\sum_{i \in I} \alpha_i(h)\Pi_i = \rho$. Therefore $q_j^D(h) = \operatorname{tr}(D_j \rho) = \operatorname{tr}(\rho D_j)$.
\end{proof}

\begin{proposition}[Agreement of the update rule]\label{prop:quantum-realization-update}
Let $\rho$ be a density operator, let $h = h(\rho)$, let $D \in \mathcal D$, and let $j \in O_D$. If $\operatorname{tr}(\rho D_j) > 0$, then the core update $\operatorname{up}_{D,j}(h)$ is defined and satisfies $\operatorname{up}_{D,j}(h) = h(\rho')$, where $\rho' := \mathcal I^D_j(\rho)/\operatorname{tr}(\mathcal I^D_j(\rho))$.
\end{proposition}

\begin{proof}
Assume $\operatorname{tr}(\rho D_j) > 0$. By Proposition~\ref{prop:quantum-realization-born}, this is equivalent to $q_j^D(h) > 0$. Hence the core update $\operatorname{up}_{D,j}(h)$ is defined.

Let $m \in I$, and set $X:=\sum_{i \in I}\alpha_i(h)\Pi_i$. By definition of the core update, $\operatorname{up}_{D,j}(h)_m=(\sum_{i \in I}\alpha_i(h)U^{D,j}(m|i))/q_j^D(h)$. Substituting the update kernel and using linearity gives the numerator $\operatorname{tr}(H_m\,\mathcal I^D_j(X))$. The SIC reconstruction formula gives $X=\rho$. Hence $\operatorname{up}_{D,j}(h)_m=\operatorname{tr}(H_m\,\mathcal I^D_j(\rho))/q_j^D(h)$.

By Proposition~\ref{prop:quantum-realization-born}, $q_j^D(h)=\operatorname{tr}(\rho D_j)$. The defining property of the instrument effect gives $\operatorname{tr}(\rho D_j)=\operatorname{tr}(\mathcal I^D_j(\rho))$. It follows that $\operatorname{up}_{D,j}(h)_m=\operatorname{tr}(H_m\rho')=h_m(\rho')$. Since this holds for every $m \in I$, we conclude that $\operatorname{up}_{D,j}(h)=h(\rho')$.
\end{proof}

Proposition~\ref{prop:quantum-realization-update} immediately yields closure of $Q_{\mathrm{SIC}}$ under one-step updates. Whenever $h = h(\rho) \in Q_{\mathrm{SIC}}$ and $q_j^D(h)>0$, the posterior reference prior $\operatorname{up}_{D,j}(h)$ belongs to $Q_{\mathrm{SIC}}$, because it is the SIC coordinate vector of the density operator $\rho'$.

\subsection{Finite-dimensional realization}\label{subsec:quantum-realization-theorem}

The kernel verification, the geometric result, and the comparison propositions combine into the following realization theorem for finite-dimensional quantum mechanics in SIC coordinates.

\begin{theorem}[Quantum realization theorem]\label{thm:quantum-realization}
Assume that a SIC exists in the chosen dimension $d$. Let $\mathcal H$, $\{H_i\}_{i \in I}$, the POVM family $\mathcal D$, and the instrument family $\mathcal I^D$ be as fixed in Subsection~\ref{subsec:quantum-realization-sic-reference-data}. Define $B^D(j|i)$, $U^{D,j}(m|i)$, and $Q_{\mathrm{SIC}}$ as in Subsections~\ref{subsec:quantum-realization-kernels} and \ref{subsec:quantum-admissible-prior-space}. Then $(Q_{\mathrm{SIC}},\mathcal D,B,U)$ is a well-formed core datum. Moreover, $Q_{\mathrm{SIC}}$ satisfies the qplex geometry conditions. Finally, for every density operator $\rho$, every $D \in \mathcal D$, and every $j \in O_D$, the core probability $q_j^D(h(\rho))$ equals the quantum Born probability $\operatorname{tr}(\rho D_j)$, and whenever this quantity is positive the core update $\operatorname{up}_{D,j}(h(\rho))$ coincides with the SIC coordinate vector of the posterior density operator $\mathcal I^D_j(\rho)/\operatorname{tr}(\mathcal I^D_j(\rho))$.
\end{theorem}

\begin{proof}
By Proposition~\ref{prop:quantum-realization-kernels}, the families $B^D$ and $U^{D,j}$ satisfy the kernel conditions specified in Subsection~\ref{subsec:semantic-data-kernels}. Density operators exist, and every $h(\rho)$ belongs to $\Delta(I)$. Hence $Q_{\mathrm{SIC}}$ is a nonempty subset of $\Delta(I)$. By Proposition~\ref{prop:quantum-realization-born}, for every $h(\rho) \in Q_{\mathrm{SIC}}$ the core Born expressions $q_j^D(h(\rho))$ are genuine probability distributions and coincide with the quantum Born probabilities. Hence the Born admissibility condition built into well-formedness is satisfied. Finally, Proposition~\ref{prop:quantum-realization-update} shows that whenever $q_j^D(h(\rho))>0$, the update $\operatorname{up}_{D,j}(h(\rho))$ after one outcome is defined and belongs again to $Q_{\mathrm{SIC}}$. Therefore the update closure condition is also satisfied. Thus $(Q_{\mathrm{SIC}},\mathcal D,B,U)$ is a well-formed core datum.

By Proposition~\ref{prop:quantum-realization-qplex-geometry}, the admissible prior space $Q_{\mathrm{SIC}}$ satisfies the consistency bounds and the lower polar condition. Therefore, by Definition~\ref{def:qplex-geometry-conditions}, $Q_{\mathrm{SIC}}$ satisfies the qplex geometry conditions.

Propositions~\ref{prop:quantum-realization-born} and \ref{prop:quantum-realization-update} establish the remaining claims.
\end{proof}

\subsection{Effective semialgebraic realization}\label{subsec:effective-semialgebraic-quantum-data}

The quantum realization belongs to the effective class from Section~\ref{sec:metatheory} under the explicit coefficient assumptions stated below.

\begin{theorem}[Effective semialgebraic bridge theorem]\label{thm:effective-bridge}
Assume the hypotheses of Theorem~\ref{thm:quantum-realization}. Let $K \subseteq \mathbb{R}$ be an effectively presented real closed field. Assume, in addition, the following.
\begin{enumerate}
\item For each $i \in I$, the real and imaginary parts of all matrix entries of $\Pi_i$ belong to $K$ and are supplied with codes in the fixed effective presentation of $K$.
\item For each $D \in \mathcal D$ and each $j \in O_D$, the real and imaginary parts of all matrix entries of $D_j$ belong to $K$ and are supplied with codes in the fixed effective presentation of $K$.
\item For each $D \in \mathcal D$ and each $j \in O_D$, the instrument map $\mathcal I_j^D$ is given by a finite \emph{Kraus decomposition} $\mathcal I_j^D(X)=\sum_{\nu=1}^{N_{D,j}} M_{j,\nu}^D X (M_{j,\nu}^D)^*$ such that the real and imaginary parts of all matrix entries of every $M_{j,\nu}^D$ belong to $K$ and are supplied with codes in the fixed effective presentation of $K$.
\item Every constant symbol of the object language names an element of $K$ and is supplied with the code of that element in the fixed effective presentation of $K$.
\end{enumerate}
Then the quantum datum $(Q_{\mathrm{SIC}},\mathcal D,B,U)$ is effectively semialgebraic over $K$.
\end{theorem}

\begin{proof}
By Theorem~\ref{thm:quantum-realization}, $(Q_{\mathrm{SIC}},\mathcal D,B,U)$ is a well-formed core datum. It therefore remains to verify the three clauses of Definition~\ref{def:effective-semialgebraic-datum}.

We begin with clause (i), which requires every kernel value to belong to $K$ and to be supplied with a code in the fixed effective presentation of $K$. Let $D \in \mathcal D$, $j \in O_D$, and $k \in I$. Write $D_j=D_j^{\mathrm{re}}+iD_j^{\mathrm{im}}$ and $\Pi_k=\Pi_k^{\mathrm{re}}+i\Pi_k^{\mathrm{im}}$, where these four real matrices have entries in $K$. Then $B^D(j|k)=\operatorname{tr}(D_j\Pi_k)$ is a real number obtained from the real and imaginary parts of the entries of $D_j$ and $\Pi_k$ by the field operations in $K$. Hence $B^D(j|k) \in K$.

Next let $m \in I$ as well. By assumption, $\mathcal I_j^D(X)=\sum_{\nu=1}^{N_{D,j}} M_{j,\nu}^D X (M_{j,\nu}^D)^*$. Therefore $U^{D,j}(m|k)=\sum_{\nu=1}^{N_{D,j}} \operatorname{tr}(H_m M_{j,\nu}^D \Pi_k (M_{j,\nu}^D)^*)$. Each summand is a real number obtained from the real and imaginary parts of the entries of $H_m$, $M_{j,\nu}^D$, and $\Pi_k$ by the field operations in $K$. Since $H_m=\Pi_m/d$, all these real and imaginary parts lie in $K$. Hence each summand lies in $K$, and so does their finite sum. Thus $U^{D,j}(m|k) \in K$ for all $D,j,m,k$.

Because the matrix entries are supplied with codes and the field operations of $K$ are recursive, codes for all the kernel values $B^D(j|i)$ and $U^{D,j}(m|i)$ can also be computed effectively. Thus clause (i) of Definition~\ref{def:effective-semialgebraic-datum} holds.

Assumption (4) verifies clause (ii) of Definition~\ref{def:effective-semialgebraic-datum}.

It remains to produce a quantifier-free formula defining $Q_{\mathrm{SIC}}$. For $h=(h_k)_{k \in I} \in \mathbb{R}^{d^2}$, write $\rho(h)=R(h)+iS(h)$, where $R(h)$ and $S(h)$ are the real and imaginary parts of $\rho(h)$. Since the coefficients $\bigl((d+1)h_k-(1/d)\bigr)$ are real and the real and imaginary parts of the entries of the $\Pi_k$ lie in $K$, every entry of $R(h)$ and every entry of $S(h)$ is an affine linear expression in the coordinates $h_k$ with coefficients in $K$.

Define the real symmetric $2d \times 2d$ matrix $A(h)$ by $A(h):=\begin{pmatrix} R(h) & -S(h) \\ S(h) & R(h) \end{pmatrix}$. We claim that $\rho(h)$ is positive semidefinite if and only if $A(h)$ is positive semidefinite. Indeed, let $x,y \in \mathbb{R}^d$ and put $v:=x+iy \in \mathbb{C}^d$. Since $\rho(h)$ is Hermitian, $R(h)^T=R(h)$ and $S(h)^T=-S(h)$. Expanding both quadratic forms gives $x^T R(h)x+y^T R(h)y-2x^T S(h)y$, and hence $v^* \rho(h) v = \begin{pmatrix} x \\ y \end{pmatrix}^{\!T} A(h) \begin{pmatrix} x \\ y \end{pmatrix}$. Therefore $\rho(h)$ is positive semidefinite if and only if $A(h)$ is positive semidefinite.

By Lemma~\ref{lem:sic-image-characterization}, one has $h \in Q_{\mathrm{SIC}}$ if and only if $h \in \Delta(I)$ and $\rho(h)$ is positive semidefinite. The equivalence established above allows $\rho(h)$ to be replaced by $A(h)$ in this condition. Since $A(h)$ is a real symmetric matrix, positive semidefiniteness is equivalent to the nonnegativity of all principal minors. Each principal minor of $A(h)$ is a polynomial in the coordinates $h_i$ with coefficients in $K$, because the entries of $A(h)$ are affine linear over $K$.

Accordingly, define $\Theta_{Q_{\mathrm{SIC}}}(\bar{x})$ to be the conjunction of the simplex conditions $x_i \ge 0$ for all $i \in I$, $\sum_{i \in I} x_i = 1$, and the finitely many inequalities asserting that every principal minor of $A(\bar{x})$ is nonnegative. This is a quantifier-free formula over $\Sigma_K$, and by the argument above it defines exactly $Q_{\mathrm{SIC}}$.

Finally, the supplied codes for the entries of the matrices $\Pi_i$, together with the effective field operations of $K$, allow the coefficients of all those polynomials to be computed effectively. Hence the formula $\Theta_{Q_{\mathrm{SIC}}}$ can also be constructed effectively. Therefore $Q_{\mathrm{SIC}}$ admits an effective semialgebraic presentation over $K$.

This verifies clause (iii) and completes the proof.
\end{proof}

\medskip

\begin{remark}[Coefficient fields for algebraic data]\label{rem:coefficient-fields-algebraic-data}
When the SIC projectors, POVM effects, and Kraus operators are given by codes for algebraic real and imaginary parts, one may take $K$ to be the real closed field of real algebraic numbers with its standard effective presentation. More generally, the theorem applies to any effectively presented real closed field that contains the finitely many real coefficients needed to describe the chosen quantum data, together with every real number named by a constant symbol of the language. Codes for all these elements must be supplied in the fixed effective presentation of $K$.
\end{remark}

This theorem places the quantum realization within the scope of the recursive metatheorem from Section~\ref{sec:metatheory}. The corresponding corollary follows.

\begin{corollary}[Recursive metatheory for effective quantum realizations]\label{cor:quantum-realization-recursive-completeness}
Under the assumptions of Theorem~\ref{thm:effective-bridge}, the recursive soundness, completeness, and decidability results of Section~\ref{sec:metatheory} apply to the quantum datum $(Q_{\mathrm{SIC}},\mathcal D,B,U)$. In particular, validity in the full QBist language over that datum is decidable.
\end{corollary}

\begin{proof}
By Theorem~\ref{thm:effective-bridge}, the datum $(Q_{\mathrm{SIC}},\mathcal D,B,U)$ is effectively semialgebraic over $K$. Therefore Theorem~\ref{thm:recursive-completeness} and Corollary~\ref{cor:decidability} apply.
\end{proof}

\subsection{A qubit realization}\label{subsec:qubit-realization}

The preceding framework becomes more concrete in the smallest nontrivial quantum case, namely $d=2$. The first example treats a single projective measurement. The second introduces an incompatible projective measurement to exhibit order dependence.

\medskip

\begin{example}[Guarded histories in a qubit SIC realization]\label{ex:qubit-sic-history}
Set $d=2$. Let $\vec{\sigma}=(\sigma_x,\sigma_y,\sigma_z)$ be the Pauli triple. Let $n_1=(1/\sqrt{3})(1,1,1)$, $n_2=(1/\sqrt{3})(1,-1,-1)$, $n_3=(1/\sqrt{3})(-1,1,-1)$, and $n_4=(1/\sqrt{3})(-1,-1,1)$. Define $\Pi_i:=(1/2)(I_{\mathcal H}+n_i\cdot\vec{\sigma})$ for $i=1,2,3,4$, and let $H_i:=(1/2)\Pi_i$. Each $\Pi_i$ is a projection with rank one. Then $\{H_i\}_{i=1}^4$ is a qubit SIC POVM.

First take the actual measurement $D=\{P_0,P_1\}$, where $P_0:=|0\rangle\langle 0|$ and $P_1:=|1\rangle\langle 1|$. The corresponding Born kernel is $B^D(j|i):=\operatorname{tr}(P_j\Pi_i)$. Since the $z$ coordinates of $n_1,n_4$ are $1/\sqrt{3}$ and those of $n_2,n_3$ are $-1/\sqrt{3}$, it follows that $B^D(0|1)=B^D(0|4)=(1/2)(1+1/\sqrt{3})$ and $B^D(0|2)=B^D(0|3)=(1/2)(1-1/\sqrt{3})$. Also, $B^D(1|i)=1-B^D(0|i)$ for each $i$.

Consider the maximally mixed state $\rho_*:=(1/2)I_{\mathcal H}$. Its SIC coordinates are $h^*:=h(\rho_*)=(1/4,1/4,1/4,1/4)$. Since $d=2$, the affine coefficients are $\alpha_i(h^*)=(d+1)h_i^*-1/d=3\cdot1/4-1/2=1/4$. Therefore $q_0^D(h^*)=\sum_{i=1}^4 \alpha_i(h^*)\,B^D(0|i)=1/2$ and $q_1^D(h^*)=1/2$.

Choose the L\"uders instrument $\mathcal I^D_j(X):=P_jXP_j$ for $j=0,1$. Then the update kernel is $U^{D,j}(m|i):=\operatorname{tr}(H_m\,\mathcal I^D_j(\Pi_i))$. Because $P_j\Pi_iP_j=\operatorname{tr}(P_j\Pi_i)\,P_j=B^D(j|i)\,P_j$, we obtain $U^{D,j}(m|i)=(1/2)B^D(j|i)B^D(j|m)$.

The posterior reference prior after outcome $0$ satisfies $\operatorname{up}_{D,0}(h^*)_m=(1/2)B^D(0|m)$. Its first and fourth coordinates are $(1/4)(1+1/\sqrt{3})$, while its second and third coordinates are $(1/4)(1-1/\sqrt{3})$.

Similarly, the first and fourth coordinates of $\operatorname{up}_{D,1}(h^*)$ are $(1/4)(1-1/\sqrt{3})$, while its second and third coordinates are $(1/4)(1+1/\sqrt{3})$. These vectors are the SIC coordinate representations of the posterior density operators $P_0$ and $P_1$.

Now let $\pi:=\varepsilon;(D,0)$. Then $r_\pi(h^*)=q_0^D(h^*)=1/2$. After this history, the posterior density operator is $P_0$, so a second performance of the same measurement satisfies $q_0^{D,\pi}(h^*)=1$ and $q_1^{D,\pi}(h^*)=0$. Accordingly, $r_{\pi;(D,0)}(h^*)=r_\pi(h^*)\,q_0^{D,\pi}(h^*)=1/2$ and $r_{\pi;(D,1)}(h^*)=r_\pi(h^*)\,q_1^{D,\pi}(h^*)=0$. By Lemma~\ref{lem:definability}, $\operatorname{up}_{\pi;(D,1)}(h^*)$ is therefore undefined, while $\operatorname{up}_{\pi;(D,0)}(h^*)$ is defined.

Let $s_*:=(h^*,\bot)$. Since $r_{\pi;(D,1)}(h^*)=0$, one has $s_* \models [\pi;(D,1)]\varphi$ for every formula $\varphi$. Also, $s_* \models \langle \pi;(D,0)\rangle \top$ because $r_{\pi;(D,0)}(h^*)>0$. Moreover, $s_* \models [D:0]\mathsf{Last}(D,0)$, which makes the record component of the semantics explicit.

The guarded formulas from Subsection~\ref{subsec:guarded-examples} become concrete in this example. Since a first outcome $0$ is executable and produces the posterior density operator $P_0$, one has $s_* \models \langle D:0\rangle(q_0^D-(3/4)\ge0)$ and $s_* \models [D:0]\neg\langle D:1\rangle\top$. The first formula states that outcome $0$ is currently possible and that, after it occurs, the agent assigns probability at least $3/4$ to the same outcome in an immediate repetition. The second excludes outcome $1$ from such a repetition.

Moreover, $s_* \not\models \langle D:1\rangle(q_0^D-(3/4)\ge0)$. Here $D:1$ is executable, since $q_1^D(h^*)=1/2$. The failure comes from the posterior commitment after outcome $1$, where the repeated probability of outcome $0$ is $0$.

The same calculation also gives a concrete instance of a mixed reference and actual outcome assertion. Since $\operatorname{up}_{D,0}(h^*)_1=(1/4)(1+1/\sqrt{3})>1/4=h_1^*$ and $q_1^D(h^*)=1/2>0=q_1^{D,\varepsilon;(D,0)}(h^*)$, one has $s_* \models \langle D:0\rangle\top\wedge(h_1^{\varepsilon;(D,0)}-h_1>0)\wedge(q_1^D-q_1^{D,\varepsilon;(D,0)}>0)$. This formula says that the executable outcome $(D,0)$ increases the probability assigned to the first SIC reference outcome while decreasing, in fact eliminating, the posterior probability of the alternative outcome $(D,1)$.

\end{example}

\medskip

\begin{example}[Measurement order in a qubit SIC realization]\label{ex:qubit-order-dependence}
Retain the qubit SIC realization, the maximally mixed prior $h^*$, the logical state $s_*$, and the measurement $D$ from Example~\ref{ex:qubit-sic-history}. Consider also the actual measurement $E=\{P_+,P_-\}$. Let $|+\rangle:=(1/\sqrt{2})(|0\rangle+|1\rangle)$ and $|-\rangle:=(1/\sqrt{2})(|0\rangle-|1\rangle)$, and put $P_+:=|+\rangle\langle +|$ and $P_-:=|-\rangle\langle -|$. Equip this measurement with its L\"uders instrument.

From the maximally mixed state one has $q_+^E(h^*)=1/2$. From the state $P_0$, the outcome $+$ of $E$ has probability $1/2$, and from the state $P_+$, the outcome $0$ of $D$ has probability $1/2$. Hence the histories $\pi_{DE}:=\varepsilon;(D,0);(E,+)$ and $\pi_{ED}:=\varepsilon;(E,+);(D,0)$ are both executable at $h^*$, with $r_{\pi_{DE}}(h^*)=r_{\pi_{ED}}(h^*)=1/4$.

The two orders nevertheless generate different posterior commitments. The posterior density operator after $\pi_{DE}$ is $P_+$, so $q_0^{D,\pi_{DE}}(h^*)=1/2$. The posterior density operator after $\pi_{ED}$ is $P_0$, so $q_0^{D,\pi_{ED}}(h^*)=1$. Therefore $s_* \models \langle \pi_{DE}\rangle\top\wedge\langle \pi_{ED}\rangle\top\wedge(q_0^{D,\pi_{ED}}-q_0^{D,\pi_{DE}}>0)$. This is a concrete instance of the formula in Example~\ref{ex:guarded-order-dependence}. Both histories are executable, but they induce different posterior probabilities for outcome $0$ in a subsequent performance of $D$.
\end{example}

\section{Concluding remarks}\label{sec:concluding-remarks}

\subsection{Discussion}\label{subsec:discussion}

The well-formed core datum provides a common semantic basis for current probability assignments, executable histories, and posterior commitments. Its assumptions comprise the kernel, probability, and closure conditions required by guarded updating. Stronger geometric and Hilbert space conditions can therefore be added at the realization stage without changing the language or its semantics. This layered formulation also permits comparisons between quantum and nonquantum realizations within the same framework.

The global reduction theorem shows that, over a fixed well-formed core datum, every formula with dynamic operators is equivalent to a formula in the fragment without dynamic operators. For each fixed effectively semialgebraic datum, the first-order translation then yields a recursive calculus that is sound and complete, together with a decision procedure. These results characterize the formulas valid over that datum.

Assuming a SIC in the chosen dimension, the realization establishes that standard finite-dimensional quantum mechanics satisfies the abstract requirements. POVMs determine the Born kernels, quantum instruments determine the update kernels, and the SIC image satisfies the consistency bounds and lower polar condition. The effective bridge theorem gives concrete coefficient field assumptions under which this realization is covered by the recursive metatheory. A remaining structural problem is the classification of abstract well-formed core data admitting Hilbert space realizations, together with the geometric or dynamic principles governing realizability.

The decidability theorem applies to each fixed effectively semialgebraic datum. The dynamic operators are indexed by concrete finite histories, which permits their elimination by the reduction developed above. Extensions involving iteration or nondeterministic program operations would require a separate metatheory. The decision procedure is obtained through first-order reasoning over real closed fields and establishes decidability at the structural level, while questions of computational efficiency remain outside the present analysis.

\subsection{Further directions}\label{subsec:further-directions}

The framework suggests several natural lines of inquiry.

One direction is geometric. Since the SIC realization yields an admissible prior space satisfying the qplex geometry conditions, further work can seek additional constraints characterizing the spaces that arise from Hilbert space quantum theory. Such constraints could concern the global organization of extreme points, the relation between the affine Born map and the geometry of admissible prior spaces, or the connection between state space geometry and update closure.

A second direction is dynamic. The present semantics suggests a more detailed study of the algebra of iterated updates, compatibility conditions between Born kernels and update kernels, and dynamic principles that may distinguish Hilbert space realizations from other realizations satisfying the qplex geometry conditions.

A third direction concerns proof theory. Relevant topics include the classification of concrete QBist realizations satisfying the effective hypotheses, the development of calculi formulated directly in the object language that internalize reasoning over real closed fields, and the proof-theoretic consequences of special QBist constraints.

\section*{Funding}
This work was supported by Japan Society for the Promotion of Science (JSPS) KAKENHI [Grant Number JP24K03372].

\end{document}